\documentclass[11pt,a4paper]{article}

\usepackage[left=0.9in, right=0.9in, top=1in, bottom=1in]{geometry}

\usepackage{amssymb}
\usepackage{amsfonts}
\usepackage{amsmath}
\usepackage{amsthm}
\usepackage{graphicx}
\usepackage[font=small]{caption}
\usepackage{subcaption}
\usepackage{bm}
\usepackage[ruled,noend]{algorithm2e}
\usepackage{algpseudocode} 
\usepackage{varwidth}
\usepackage{placeins}
\usepackage[utf8]{inputenc}
\usepackage{mathtools} 
\usepackage{hyperref}
\usepackage{booktabs}
\usepackage{verbatim}
\usepackage{enumitem}
\usepackage{soul}
\usepackage{scalerel}

\newcommand{\footremember}[2]{%
    \footnote{#2}
    \newcounter{#1}
    \setcounter{#1}{\value{footnote}}%
}

\newcount\Comments
\usepackage{color}
\definecolor{darkgreen}{rgb}{0,0.5,0}
\newcommand{\kibitz}[2]{\ifnum\Comments=1{\color{#1}{#2}}\fi}

\newcommand{\ignore}[1]{}
\newtheorem{theorem}{Theorem}

\newtheorem*{theorem*}{Theorem} 

\newtheorem{lemma}{Lemma}
\newtheorem{proposition}{Proposition}
\newtheorem{definition}{Definition}

\algtext*{EndFor}
\algtext*{EndIf}

\title{\vspace{12pt}Auctions between Regret-Minimizing Agents\thanks{This work has been published in Proceedings of the ACM Web Conference 2022 (WWW'22). \href{https://doi.org/10.1145/3485447.3512055}{https://doi.org/10.1145/3485447.3512055}.}}

\author{
Yoav Kolumbus\footremember{HUJI1}{The Hebrew University of Jerusalem, Israel. Email: yoav.kolumbus@mail.huji.ac.il.}%
\and Noam Nisan\footremember{HUJI2}{The Hebrew University of Jerusalem, Israel. Email: noam@cs.huji.ac.il.}
	}
\date{}

\begin{document}
\maketitle
\vspace{12pt}
\begin{abstract}
We analyze a scenario in which software agents implemented as regret-minimizing algorithms engage in a repeated auction on behalf of their users. We study first-price and second-price auctions, as well as their generalized versions (e.g., as those used for ad auctions). Using both theoretical analysis and simulations, we show that, surprisingly, in second-price auctions the players have incentives to misreport their true valuations to their own learning agents, while in first-price auctions it is a dominant strategy for all players to truthfully report their valuations to their agents.
\end{abstract}

\bibliographystyle{splncs03}

\section{Introduction}
This paper deals with the following common type of scenario: several users engage in a repeated online auction, where each of them is assisted by a learning agent.  A typical example is advertisers that compete for ad slots: typically, each of these advertisers enters his key parameters into some advertiser-facing website, and then this website's ``agent'' participates on the advertiser's behalf in a sequence of auctions for ad slots.  Often, the auction's platform designer provides this agent as its advertiser-facing user interface.  In cases where the platform's agent does not optimize sufficiently well for the advertiser (but rather, say, for the auctioneer), one would expect some other company to provide a better (for the advertiser) agent.  A typical learning algorithm of this type will have at its core some regret-minimization algorithm \cite{hart2000simple,blum2007learning}, such as ``multiplicative weights'' (see \cite{arora2012multiplicative} and references therein),  ``online gradient descent'' \cite{zinkevich2003online}, or some variant of fictitious play \cite{brown1951iterative,robinson1951iterative}, such as ``follow the perturbed leader'' \cite{hannan1957lapproximation,kalai2005efficient}. In particular, for ad auctions, there is empirical data showing that bids are largely consistent 
with no-regret learning \cite{nekipelov2015econometrics, noti2021bid}.  

While one's first intuition would be that an auction that is well designed for ``regular'' users would also 
work well for such agent-assisted users, it turns out that this is often not the case. In this paper we 
study how basic auction formats perform in terms of revenue and economic efficiency when 
software agents are doing the repeated bidding for the users.  Specifically, we assume that each of the
human bidders enters his valuation into a regret-minimizing software agent, and then the agent engages in
a long sequence of auctions, attempting to maximize utility for the human user.  We present both 
simulation results and theoretical analysis for several types of auctions: first-price, second-price, and 
generalized first- and second-price auctions (for multiple ad slots).  We study both classic regret-minimizing
algorithms like multiplicative weights or follow the perturbed leader, as well as the consequences of the
regret-minimization property itself.  

Before presenting our results, let us say what exactly we are looking at when studying a long sequence of auctions played by software agents.  While one might hope that there is some kind of convergence to fixed bids or distributions, it is well known that regret-minimization dynamics in many games need not converge \cite{daskalakis2010learning,hart2013simple,papadimitriou2018nash,bailey2018multiplicative,bailey2021stochastic}. For example, in our context, Figure \ref{fig:GFP_symmetric_bid_dynamics} shows the dynamics of a simulation of a generalized first-price auction for two slots where the two bidders have the same value.  
(Throughout the dynamics, both agents bid nearly the same value and thus only a single bid-line is visible.) 
As one may clearly see, the dynamics do not converge.  What we wish to study is the thing that is important to the auctioneer and to the bidders: the long-term distribution of the tuple of bids. For example, Figure \ref{fig:GFP_joint_bid_dist_diagonal} shows a density map of the empirical pair of bids in this simulation.  The empirical distribution of the tuple of bids determines the average winning frequencies and prices, which in turn determine the users' utilities and the auctioneer's revenue over the whole sequence of auctions. While there may not always be a theoretical guarantee that this empirical distribution converges either, one may still theoretically analyze properties that it may have in various cases as well as study it through simulations.

It is well known that in any game and for any regret-minimizing dynamics the distribution of tuples of bids
approaches a ``coarse correlated equilibrium'' \cite{hannan1957lapproximation,hart2000simple,young2004strategic,blum2007learning} (CCE), 
but as the space of CCEs is quite wide, this often does not suffice for an analysis of parameters of interest like the revenue
or social welfare.  For example, both for the first- and second-price auctions, there exist multiple (full information) CCEs, with different revenues and different social welfare \cite{feldman2016correlated}. One of the contributions of this work, presented in Section \ref{sec:co-undominated-CCE}, is the identification of a refinement of the class of CCEs which we term {\em co-undominated},  
to which alone a large class of regret-minimizing algorithms (``mean-based'' ones, as defined in \cite{braverman2018selling}) may converge in any game. 
In some cases this refinement suffices for an analysis, and for specific regret-minimization algorithms like multiplicative weights we may get an even finer analysis.

While there has been some work on the behavior of no-regret dynamics in auctions \cite{daskalakis2016learning,balseiro2019learning,feng2020convergence,deng2021nash}, we do not know of previous analysis of the average utilities and revenue of standard no-regret dynamics in standard auction formats. 
The closest paper conceptually to ours is \cite{feng2020convergence}, which manages to analyze these for no-regret dynamics that are preceded by a long pure exploration
phase by all agents. Here we provide theoretical analysis and simulations for standard regret-minimization algorithms (without a pre-exploration phase), 
sometimes getting very different results. In a companion paper  \cite{KolumbusNisan2021manipulate} we discuss the convergence of regret-minimizing agents and  analyze the incentives of the users of such agents in a general game-theoretic setting.

\begin{figure}[!t]
\centering
\vspace{-34pt}
	\begin{subfigure}{.49\linewidth}
			\vspace{8pt}

		\includegraphics[width=1.11\linewidth]{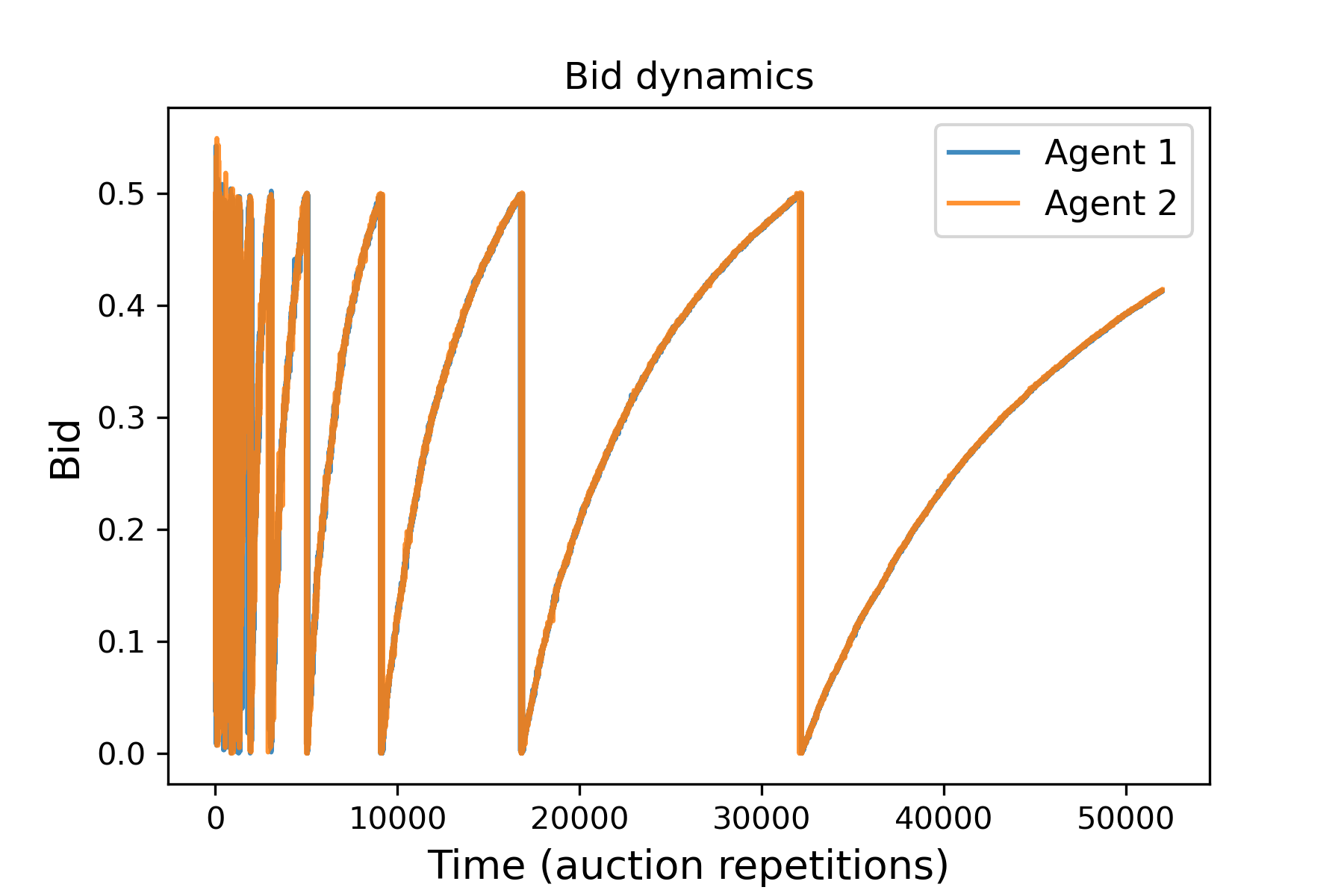}
		\caption{Bid dynamics}  
		\label{fig:GFP_symmetric_bid_dynamics}
	\end{subfigure}
	\begin{subfigure}{.49\linewidth} 
	\center
		\includegraphics[width=0.78\linewidth]{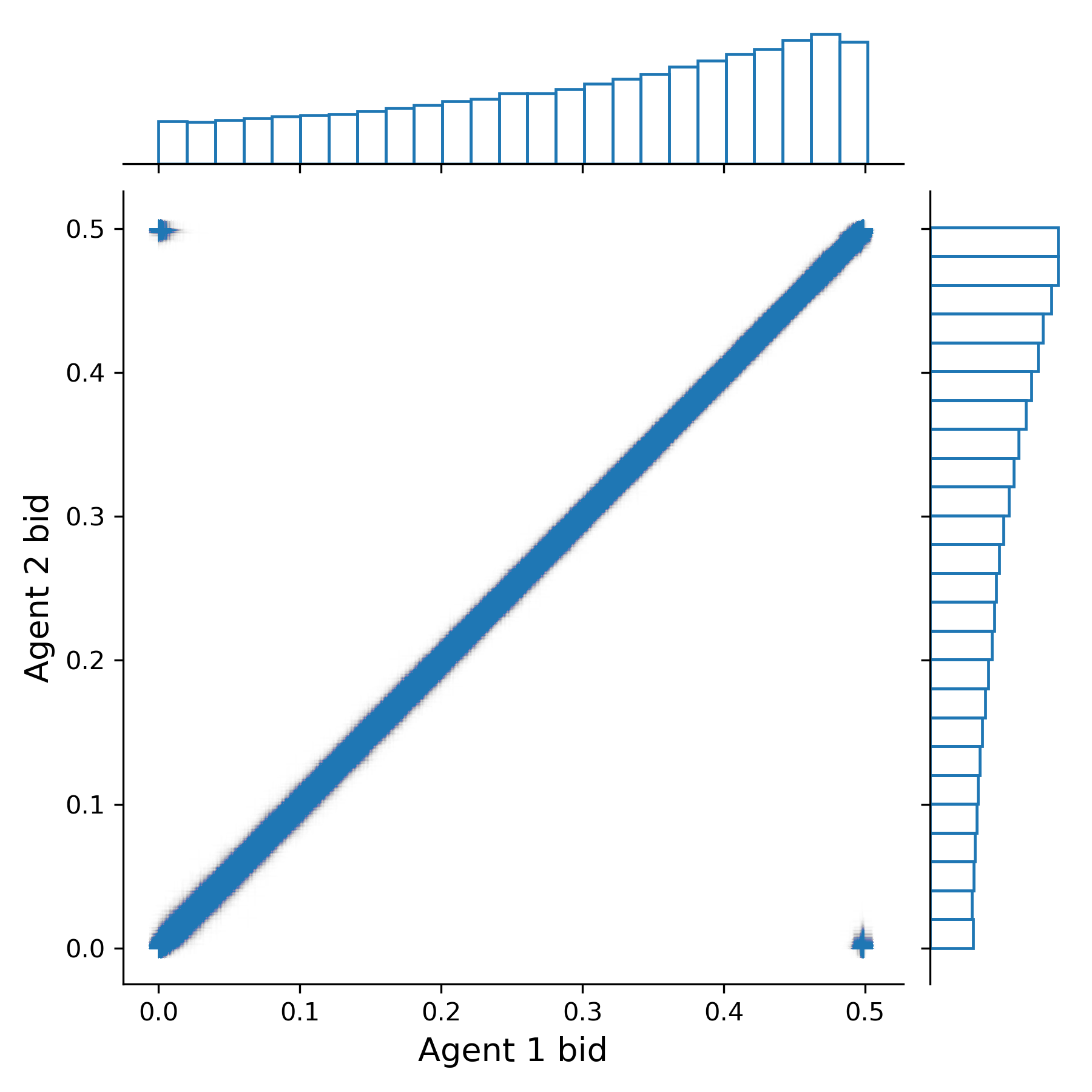}
		\caption{Joint distribution of bids}  
		\label{fig:GFP_joint_bid_dist_diagonal}
	\end{subfigure}
		\vspace{-4pt}
	\caption{{\small Dynamics and bid density in a symmetric GFP auction.}}
	\vspace{-8pt}
	\label{fig:GFP-bids}
\end{figure}

\subsection{Second-Price Auctions}
To demonstrate the difficulty, let us start by looking at the easiest-to-analyze auction: a second-price
auction of a single item among two users.  It is well known \cite{vickrey1961counterspeculation,edelman2007internet} that such an auction is 
(dominant-strategy) incentive compatible, i.e., that no bidder ever has any incentive to misreport his value to the auctioneer. Such truthful bidding will
result in an efficient allocation where the player with the higher value wins and pays a price equal to the value of the lower player.

So now let us look at our scenario,
where the two (human) users each report their value to a software agent (one agent for each player),  
and then the two (software) agents repeatedly engage in auctions, each agent attempting to optimize for its
own user according to the value reported to it.  (We assume that each auction in the sequence is for the same type of item, and 
that the values of the players for this type of item are fixed and additive over the items won.)  What do we expect to happen in this
case? Specifically: \emph{Is it in the best interest of users to report their true values
to their agents?  Will the player with the higher value always win? Will he pay a price equal to 
the value of the lower player?}

While positive answers to these three questions may seem to be trivially true, this turns out not to be 
the case, and only the second question gets a positive answer, while the answers to the first and third questions are negative.  

\begin{theorem}\label{thm:second-price-auction-MW-bid-distributions}
For two bidders with values $v>w$ with multiplicative-weights agents that engage in a sequence of second-price auctions on their behalf with discrete bid levels, the distribution of bids of the high player converges to uniform in the range $(w,v]$, while the 
distribution of the bids of the low player converges to a distribution that has full support on $[0,w]$.  In
particular, in the limit, the high player always wins and pays an average price that is strictly less than the second price.
\end{theorem}

Figure \ref{fig:SecondPriceV1W05BidDynamics} shows the dynamics of multiplicative-weights learning agents in a 
second-price auction, where the values of the two agents are $v=1$ and $w=0.5$, and Figure
\ref{fig:SecondPriceV1W05BidDensity} shows the long-term distribution of bids of each of the two agents that indeed 
behave as the theorem states: the high-value agent always wins -- except for a small number of auctions in the initial phase -- and the average price he pays is about $0.27$ (which is well below $w$). This is very much in contrast to the results of \cite{feng2020convergence}
that show that if we start with a sufficiently long pure exploration phase in which all agents bid uniformly, then dominant strategies are indeed reached.  
While the proved and observed lack of convergence to the dominant strategies may be surprising, the explanation is simple: while the dominant strategies do form a Nash equilibrium of the game, there are other (full-information) Nash equilibria of this game where the first player bids arbitrarily 
above the second price and the second player underbids arbitrarily.\footnote{
These equilibria are distinct from the known ``overbidding equilibria'' where the low bidder
bids above the high-bidder's value and occur even if we completely disallow our learning agents to 
overbid as we do in our simulations.}  
Regret dynamics do not necessarily choose the former and so we do get some distribution over the ``low revenue'' equilibria.

\begin{figure}[!t]
\centering
\vspace{-24pt}
	\begin{subfigure}{.32\linewidth}

		\includegraphics[width=1.13\linewidth]{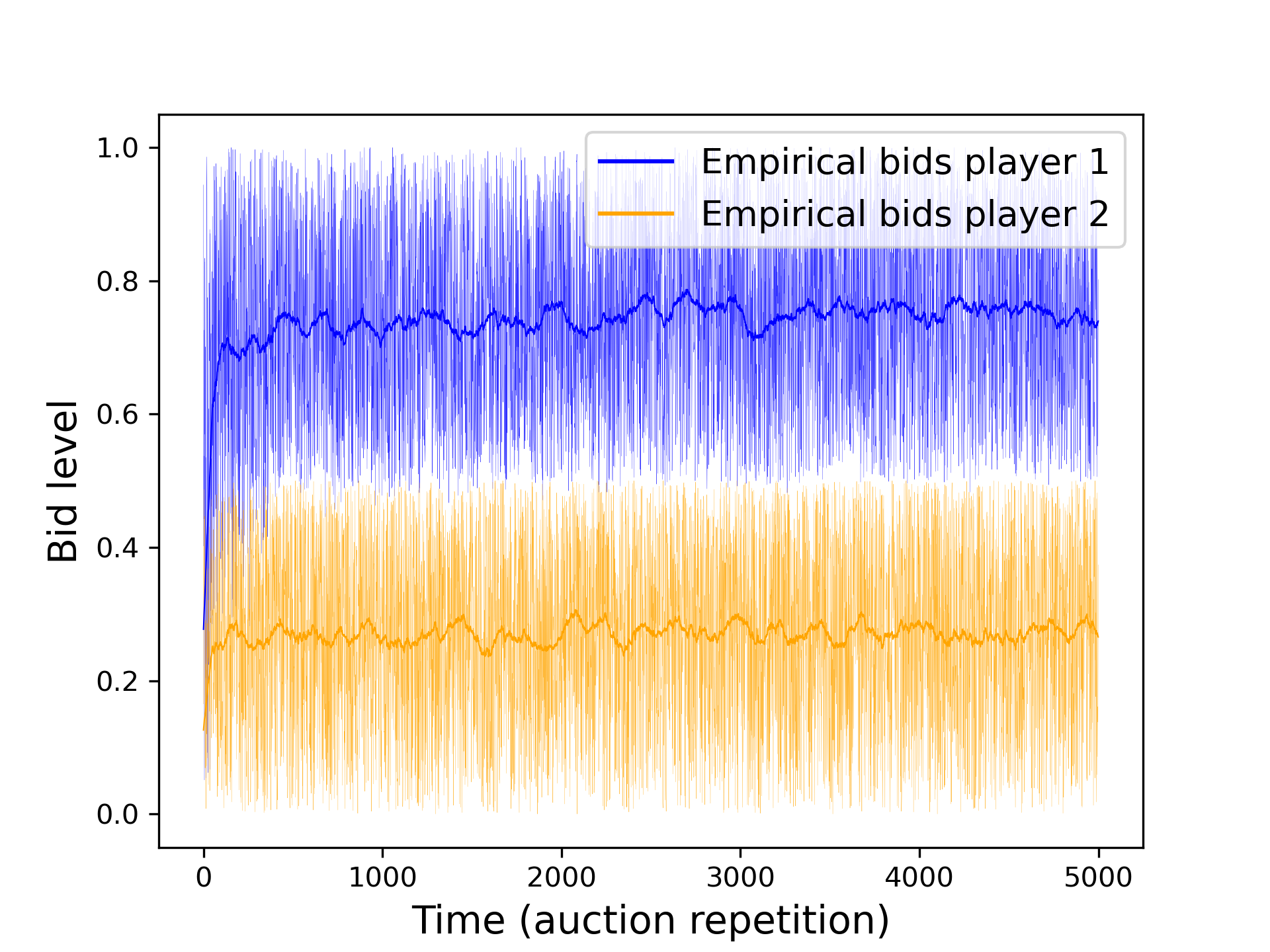}
		\caption{Bid dynamics}  
		\label{fig:SecondPriceV1W05BidDynamics}
	\end{subfigure}
	\begin{subfigure}{.32\linewidth} 
	\center
		\includegraphics[width=1.13\linewidth]{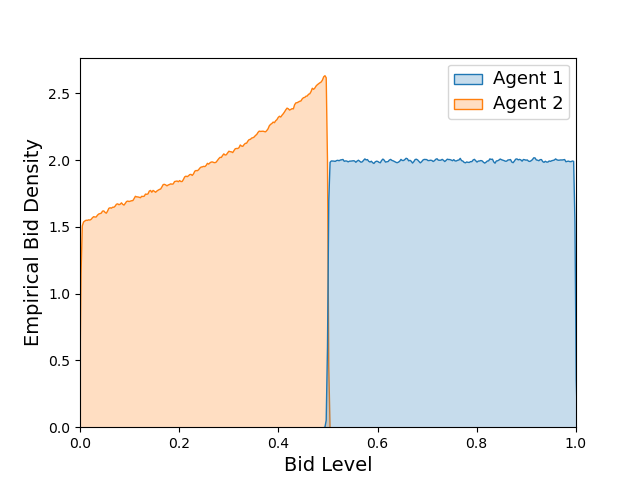}
		\caption{Bid distributions}  
		\label{fig:SecondPriceV1W05BidDensity}
	\end{subfigure}
	\begin{subfigure}{.32\linewidth}
		\includegraphics[width=1.13\linewidth]{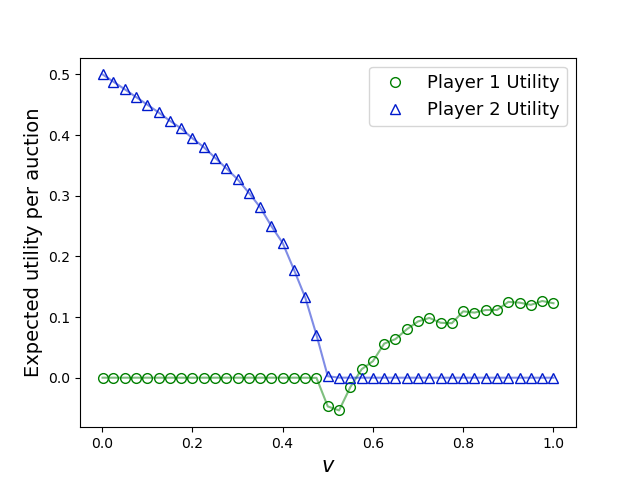}
		\caption{Value manipulation}
		\label{fig:SP_gains_from_manipulation}
	\end{subfigure}
	
	\caption{{\small Second-price auctions. Figure \ref{fig:SecondPriceV1W05BidDynamics}: Dynamics of multiplicative-weights agents, where agent $1$ has value $1$ and agent $2$ has value $0.5$. The solid lines show running-window averages of the bids of each agent over a window of $100$ auctions. 
	Figure \ref{fig:SecondPriceV1W05BidDensity}: Bid distribution in $5\cdot 10^6$ auctions for agent $1$ (in blue) and agent $2$ (in orange). 
Figure \ref{fig:SP_gains_from_manipulation}: Gains from value manipulation in a second-price auction where player $1$ has value $0.4$ and player $2$ has value $0.5$. The figure shows the utility for player $1$ from misreporting his value to his agent as $v$  for different $v$ values.
	}}
	\vspace{-12pt}
	\label{fig:second-price-suctions}
\end{figure}

The fact that the average price paid 
is well below the second price has immediate implications for a human user that has a value that lies 
between the average price in these dynamics and the actual second price. Such a user will gain from misreporting (exaggerating) his value {\em to his own agent}.  
In this example, a player with value, say, $v=0.4$ that bids truthfully against a second player with value $w=0.5$ will always lose, whereas misreporting his value as $v=1$ will lead to these simulated dynamics (as in Figure \ref{fig:second-price-suctions}), making him always win and achieve a strictly positive utility on average. One can continue this analysis, viewing the reports of the users to their own agents as defining a 
``meta-game'' between the users, and study the reached equilibrium. In a companion paper \cite{KolumbusNisan2021manipulate}, we define and study this type of 
``meta-game''  in a general context, and in Section \ref{sec:second-price} we analyze it for the second-price 
auction.\footnote{Spoiler: there are multiple possible equilibria in which the high player exaggerates his value as much as possible 
and the low player reports anything below a maximum value that is less than twice the true value of the high player.}

In Appendix \ref{sec:appendix-second-price} we also demonstrate that these types of phenomena happen not only with 
multiplicative-weights algorithms, but also with other classic regret-minimization learning algorithms such as 
``follow the perturbed leader'' (FTPL). We also show in Section \ref{sec:second-price} that the same phenomena happen in a generalized second-price (GSP) auction 
similar to the auctions used by search engines to sell ad slots in search result web pages.  
One may well wonder whether the advertiser-optimizing software supplied by the search engines indeed optimizes for the advertiser 
(in the sense of regret minimization), and whether the search engines lose revenue by using such a generalized second-price
auction.\footnote{Practically, as the number of advertisers may be large, with some distribution
on the participants' values in each auction, the revenue loss should be smaller than in our simulations. Still, it is far from clear that this loss is small.}

\subsection{First-Price Auctions}
We now turn our attention to basic first-price auctions (for a single good).   What do we expect to happen
in a repeated first-price auction between two regret-minimizing algorithms with values $v \geq w$?
Our basic intuition may be that they reach some kind of (full information) equilibrium
of the first-price auction.  It is well known that the only pure Nash equilibrium is where the low bidder
bids his value (up to an $\epsilon$) and the high bidder bids $\epsilon$ more than that.  
This means that we would expect the high player to win (almost always) and pay (close to) the second price.  Mixed Nash and 
correlated equilibria of the first-price auction are also known to yield the same outcome 
(namely, the high player always wins and pays the second price) \cite{feldman2016correlated}. Figure \ref{fig:first-price-suctions-dynamics} strengthens this intuition by showing the dynamics of first-price auctions with values $v=1$ and $w=0.5$ (Figure \ref{fig:FirstPrice-non-symmetric}) and with symmetric values $v=w=1$ (Figure \ref{fig:FirstPrice-symmetric}) between two multiplicative-weights algorithms that indeed converge to this outcome.\footnote{It should be noted, though, that this convergence in the symmetric case is very slow, perhaps the slowest that we have seen in our simulations.}

\begin{figure}[!t]
\centering
\vspace{-24pt}

	\begin{subfigure}{.49\linewidth} 
	\center
		\includegraphics[width=1.0\linewidth]{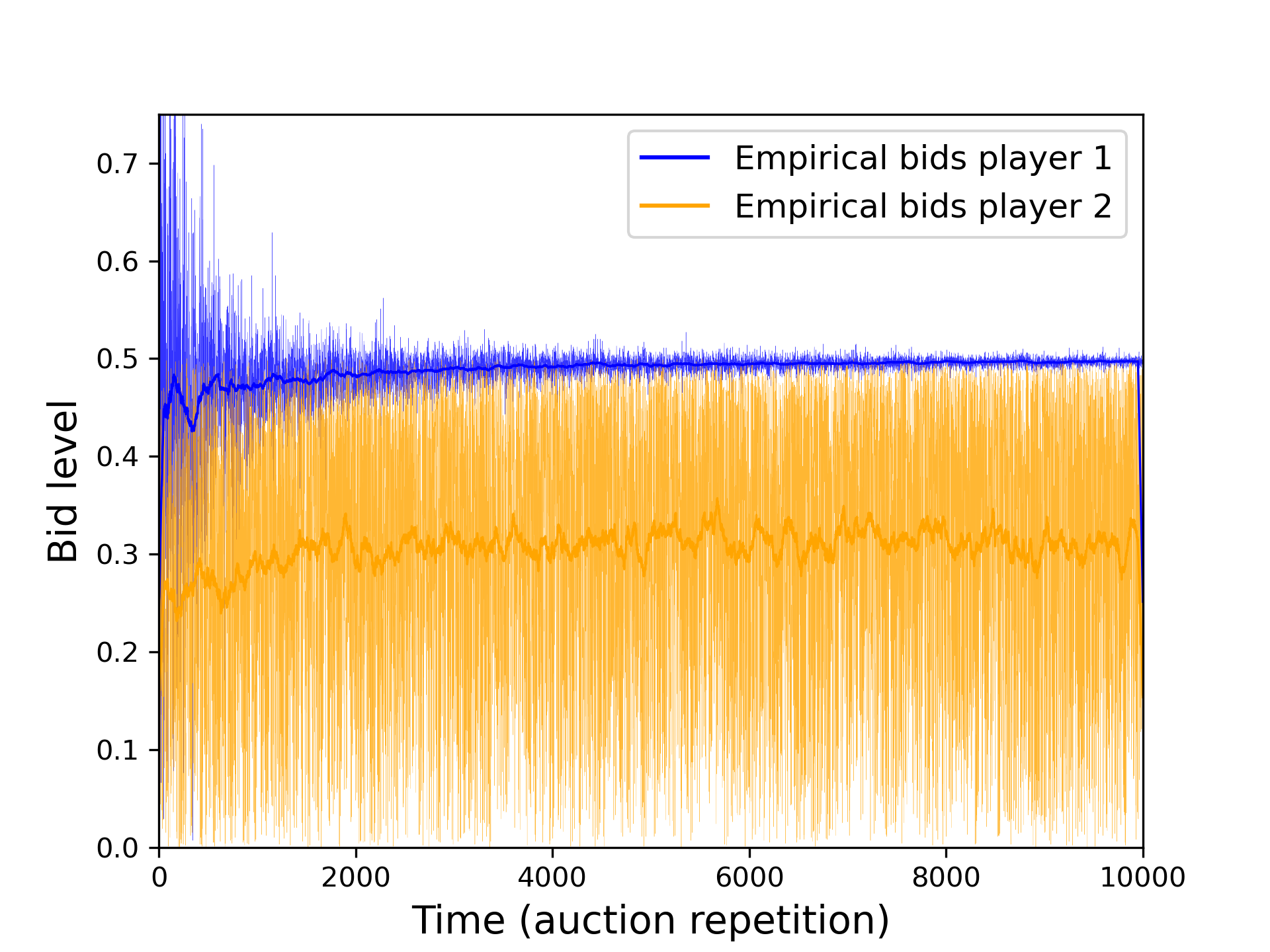}
		\caption{$v=1$, $w=0.5$}  
		\label{fig:FirstPrice-non-symmetric}
	\end{subfigure}
		\begin{subfigure}{.49\linewidth}
		\vspace{17.5pt}
		\includegraphics[width=0.894\linewidth]{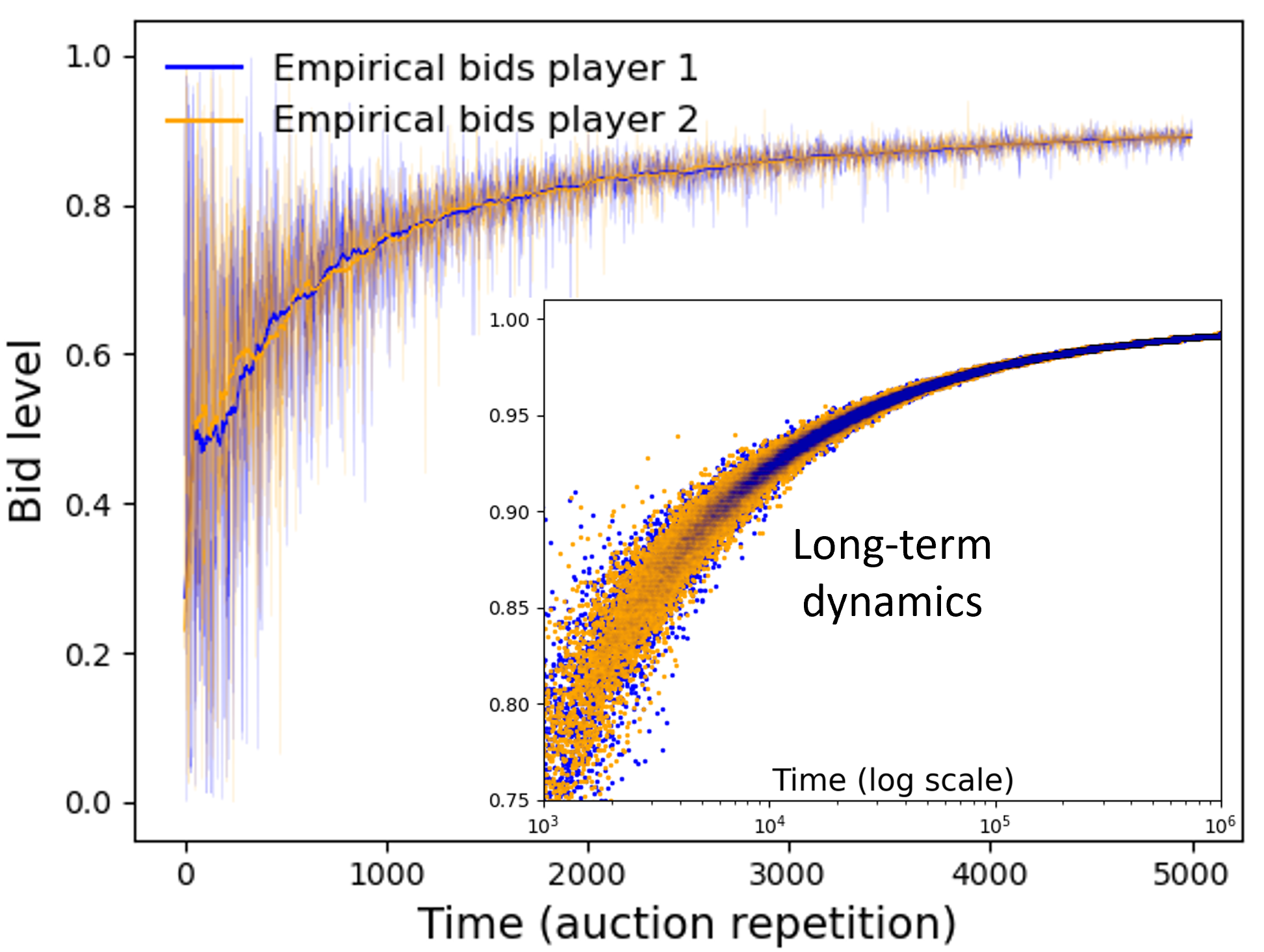}
		\caption{$v=w=1$}
			\vspace{0.5pt}
		\label{fig:FirstPrice-symmetric}
	\end{subfigure}
	\vspace{-6pt}
	\caption{{\small Dynamics of multiplicative-weights agents playing a repeated first-price auction, where agent $1$ has value $v$ and agent $2$ has value $w$. The solid lines show running-window averages of the bids of each agent over a window of $100$ auctions.	}}
	\vspace{-14pt}
	\label{fig:first-price-suctions-dynamics}
\end{figure}

The missing element in this intuition is that regret-minimization dynamics need not converge to a Nash equilibrium (or even a correlated
equilibrium) but may reach an arbitrary coarse correlated equilibrium, and it turns out that the first-price auction also has other coarse correlated equilibria \cite{feldman2016correlated}, and in many of these coarse correlated equilibria the revenue is well below the second price, and sometimes even the low player wins. While in \cite{feldman2016correlated} these coarse correlated equilibria were analyzed in a continuous setting, they also
exist with simple discrete bidding levels.  For example, consider an auction with $v=w=1$;  
the reader may verify that the distribution that gives equal weights to the following four
pairs of bids is a coarse correlated equilibrium: $\{(0,0), (0.46, 0.46), (0.64,0.64), (0.73,0.73)\}$.
It is true in general that for {\em every} CCE there exist {\em some} regret-minimizing algorithms that converge to it (\cite{monnot2017limits}, see also \cite{KolumbusNisan2021manipulate}). Specifically, in our setting, consider agents that start by playing a sequence of bids that converge to
the desired CCE (according to some fixed-in-advance schedule); each software agent proceeds this way as long as the other agent does, and in the case of any deviation reverts to a standard regret-minimizing algorithm.  Thus, regret-minimizing agents need not converge to the ``second-price outcome''  
in general.\footnote{Note that if one considers only dynamics of algorithms that minimize conditional (``internal'') regret which are known to approach the polytope of correlated equilibria \cite{foster1997calibrated,fudenberg1999conditional,hart2000simple,blum2007learning} (rather than minimizing unconditional (``external'') regret), then, since correlated equilibria of the first-price auction yield the second-price outcome \cite{feldman2016correlated}, 
the empirical average of such joint dynamics would approach the second-price outcome.}
  We are able to prove, however, that a large class of regret-minimizing algorithms, 
defined in \cite{braverman2018selling} and called ``mean-based'' algorithms, may converge only to this second-price outcome.
Informally, algorithms from this class cannot put too much probability on actions that previously performed poorly.

\begin{theorem}\label{thm:first-price-auctions}
In a repeated first-price auction between two mean-based regret-minimizing agents with possible bid
levels that are a discrete $\epsilon$-grid, if the dynamics converge to any single distribution,  
then the high player wins with probability approaching $1$ and pays an average price that converges to the second price (up to $O(\epsilon)$).
\end{theorem}

Following our simulations, we conjecture that the assumption that the dynamics converge 
(a notion defined formally in Section \ref{sec:preliminaries}) is not really required, at least for the multiplicative-weights algorithm.

As opposed to the case of the second-price auction where our results were in contrast to those 
with an additional pre-exploration phase \cite{feng2020convergence}, here we get the same type of result  
but for natural algorithms, without requiring this additional exploration phase.

A rather direct corollary of the fact that we get the second-price outcome is that from the point of view
of the human users that report their values to their regret-minimizing agents, they observe a
second-price auction.  Since this auction is incentive compatible, we get that there is never any incentive 
for a user to misreport his value to his own agent (up to a vanishing loss in the long term).  Recall
that this is in contrast to the case of the second-price auction described above.

\subsection{Generalized First-Price Auctions}
We now proceed to the more complex generalized first-price (GFP) auction.  
Specifically, we focus our attention on the following first price two-slot auction: two ``ad slots'' are being sold 
in an auction to two bidders. The ``top'' ad slot has a ``click-through rate'' of $1$, while the ``bottom'' ad slot 
has a click-through rate of $0.5$.  Each of the bidders has a value ``per click'' and so his value for the
bottom slot is exactly half of his value for the top slot.

Our first result is the analysis of the (mixed) Nash equilibrium of
this game.  While there has been much work on equilibria in related auction settings (e.g., 
\cite{chawla2013auctions, dutting2019expressiveness, feldman2016correlated,hoy2013dynamic}), 
to the best of our knowledge, this paper is the first to
successfully analyze the Nash equilibrium of 
a generalized first-price auction.

\begin{theorem}\label{thm:GFP-mixed-Nash}
Assume w.l.o.g. that $v \geq w$. In the unique Nash equilibrium of the two-slot auction described above, the players mix their bids $x,y$, respectively,
according to the following cumulative density functions with support $[0,w/2]$. 
\small
\begin{equation*}
F(x) = \frac{x}{w-x} \quad \quad \quad \quad G(y) = \left(2\frac{v}{w} -1 \right)\frac{y}{v-y}.
\end{equation*} 
\normalsize
The expected payoffs are {\small$u_1 = v/2 + (v-w)(1 - \ln(2))$}, and {\small$u_2 = w/2.$}
\end{theorem}

Running an actual simulation of multiplicative-weights agents participating in such an auction with 
$v=w=1$ shows that the agents do not get even close to any uncoupled mixing profile as this Nash equilibrium.  Figure \ref{fig:GFP_symmetric_bid_dynamics}
 shows the dynamics of a (typical) simulation. Throughout the dynamics, both agents bid nearly the same and thus only a single bid-line is visible. 
Figure \ref{fig:GFP_joint_bid_dist_diagonal} shows the density of the empirical pair of bids. The lack of convergence to a Nash equilibrium is not surprising since we only expect to reach a CCE of the game, not necessarily the Nash equilibrium.  

Most of Section \ref{sec:GFP} is devoted to studying the equilibrium that is actually reached by multiplicative-weights algorithms.
We observe via simulations that the bid distribution of the high bidder agrees with the computed distribution of the Nash equilibrium {\small $F(x) = \frac{x}{w-x}$}, but that the other player's bid is strongly correlated with the high player's bid (rather than independent of it as in the Nash equilibrium).  In fact, we see that in the achieved equilibrium the low player either bids (almost) identically to the first player or bids (close to) $0$, where the probability of the latter is proportional to the ratio between the high and low bids. As we show in Section \ref{sec:GFP}, our closed-form formulas agree well with the simulation results, but we do not have a complete theoretical analysis of this and such an analysis seems to be difficult.

\begin{figure}[!t]
\centering
\vspace{-24pt}
	\begin{subfigure}{.49\linewidth}
		\includegraphics[width=1.0\linewidth]{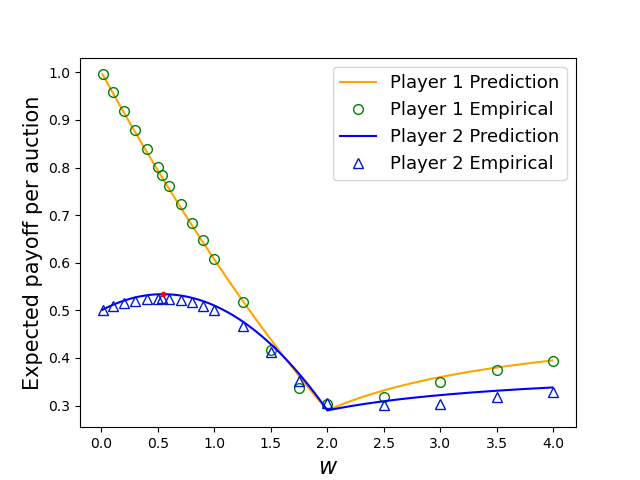}
		\caption{Player 1 plays $v=2$}
		\label{fig:w_manipulation_fixed_v=2}
	\end{subfigure}
	\begin{subfigure}{.49\linewidth} 
		\includegraphics[width=1.0\linewidth]{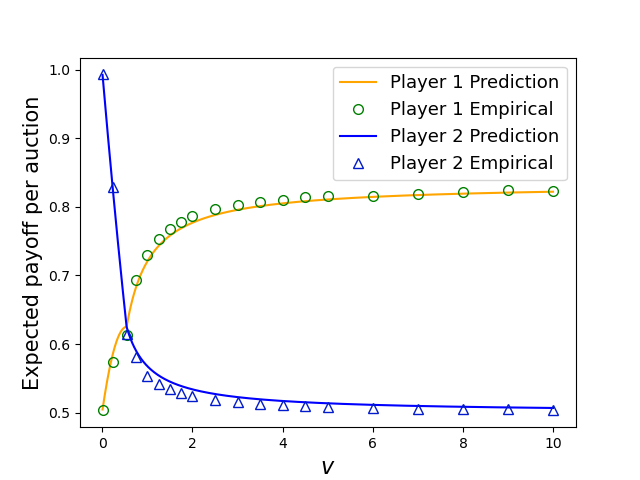}
		\caption{Player 2 plays $w\cong0.54$}  
		\label{fig:v_manipulation_fixed_w=054}
	\end{subfigure}
\vspace{-2pt}
\caption{{\small Payoffs in generalized first-price auctions according to our analytically predicted distribution compared with simulation results. The figure depicts the payoffs of the two players when one player declares a fixed value as a function of the declared value of the other player, where the declaration of player 1 is denoted by $v$ and the declaration of player 2 is denoted by $w$. }}
\label{fig:GFP-empirical-equilibrium}
\vspace{-3pt}
\end{figure}

The closed-form formulas do allow us, however, to
analytically compute the achieved winning probabilities and prices in this auction and thus also the users' payoffs from 
different reports in the ``meta-game'' where the users report their values 
to their respective agents. In particular, we see that even with true reports, the users' average utilities are (slightly) higher than in the 
Nash equilibrium, implying that their agents achieved some sort of implicit collusion.  
Demonstrating an analysis of strategic reports of the users to their agents,
Figures \ref{fig:w_manipulation_fixed_v=2} and \ref{fig:v_manipulation_fixed_w=054} show the utilities that
users with a true value of $1$ would get from strategically reporting different values to their own agent as a reply to two fixed reports ($0.54$ and $2$)\footnote{The value $2$ is used here as a proxy for ``a very high value''; the results for higher values are similar but harder to eyeball.} of the other
user.  We first see the very close agreement of the closed-form formulas with the actual simulation results.  
Second, these results demonstrate the equilibrium in the meta-game where one bidder exaggerates his bid as much as possible, while the other one underbids somewhat.  Calculating the utilities of the players in this meta-game equilibrium reveals that severe implicit collusion occurs, where the users capture $89\%$ of the total welfare, leaving only $11\%$ to the auctioneer 
(in contrast to $33\%$ to the auctioneer in the Nash equilibrium and $31\%$ to the auctioneer in the equilibrium reached for truthful reports).

\section{Further Related Work}
Regret minimization in repeated strategic interactions and in online decision problems has been extensively studied in the literatures of game theory, machine learning, and optimization. Early regret-minimizing algorithms were related to the notion of fictitious play \cite{brown1951iterative,robinson1951iterative} (a.k.a. ``follow the leader''), which in its basic form does not guarantee low regret, but whose smoothed variants, such as ``follow the perturbed leader'' (FTPL) \cite{hannan1957lapproximation,kalai2002geometric,kalai2005efficient,grigoriadis1995sublinear} or ``follow the regularized leader'' (FTRL) \cite{shalev2011online}, are known to guarantee adversarial regret of $O(\sqrt{T})$ in $T$ decisions. Another common approach to regret minimization is the 
``multiplicative-weights'' algorithm which has been developed and studied in many variants (see \cite{arora2012multiplicative} and references therein). 
Similar ideas of regret minimization apply both in settings in which players observe full feedback on their past actions and forgone actions (the experts-advice setting) and in the multi-armed bandit setting, where feedback is available at every step only for the action that was selected at that step, and in which variants of multiplicative weights (such as EXP3 \cite{auer2002finite,langford2007epoch}) and other algorithms have been developed.  

Regret minimization has also been widely studied in auctions, with two prominent lines of work focusing on the econometric problem of inferring structural parameters of the utilities of the bidders \cite{nekipelov2015econometrics,Noti2017Empirical,nisan2017quantal,gentry2018structural}, 
and on the price of anarchy in classic auction formats and in combinatorial auctions when the bidders are no-regret learners \cite{blum2008regret,roughgarden2012price,syrgkanis2013composable,caragiannis2015bounding,hartline2015no,daskalakis2016learning,roughgarden2017price}. 
Other works deal with the design of automated bidding algorithms for auctions with more complex factors, such as budget constraints \cite{balseiro2019learning,deng2021autoBidding}, or the utilization of regret minimization ideas for optimizing reserve prices \cite{roughgarden2019minimizing,mohri2014optimal,cesa2014regret}, and in \cite{noti2021bid,alaei2019response} regret-minimization algorithms were used to predict bidder behavior. Importantly, \cite{nekipelov2015econometrics,noti2021bid} provide concrete empirical evidence from large ad auction datasets showing that actual bidder behavior is consistent with no-regret learning.   

The standard first-price and second-price auctions 
have been widely studied in auction theory \cite{krishna2009auction,milgrom1987auction}. 
Their generalized versions for multiple items have been studied and applied in ad auctions (also called ``position auctions'') \cite{edelman2007internet,varian2007position,borgs2007GFPdynamics,jansen2008sponsoredSearchReview}. Most prior works mainly discussed the Nash equilibria of the auctions or their best-response dynamics, and did not study the dynamics when players use regret-minimization algorithms. One notable exception is \cite{feng2020convergence}; however, as discussed in the introduction, they analyze a different setting in which the actual play of the agents is preceded by a pure exploration phase in which all bidders bid uniformly, which leads to different convergence results than those we show for standard interaction between the algorithms. 

Finally, there is a growing body of work studying various aspects of interactions between strategic players and learning algorithms, including learning from strategic data \cite{cai2015optimum,hardt2016strategic,chen2019learning,dong2018strategic, haghtalab2020maximizing,levanon2021strategic,ghalme2021strategic}, 
Stackelberg games \cite{Jiarui2019imitativeFollower,birmpas2020optimally}, security games \cite{nguyen2019deception,nguyen2019imitative,gan2019manipulating}, 
recommendation systems \cite{tennenholtz2019rethinking,ben2018game}, and optimization against regret-minimizing opponents \cite{deng2019strategizing,braverman2018selling}. Notably, we use the formalism of mean-based learning algorithms defined in \cite{braverman2018selling}, who also study an auction setting. Their setting, however, is basically different from ours. They consider a game in which a strategic seller seeks revenue-maximizing auction strategies against a regret-minimizing buyer, while we consider fixed auction rules and study the interactions between regret-minimizing buyers and the strategic considerations of their users. Lastly, in a companion paper \cite{KolumbusNisan2021manipulate}, we formalize a ``meta-game'' model for general strategic settings in which human users select parameters to input into their learning agents that play a repeated game on their behalf, and we analyze this model for simple games with general regret-minimizing agents. Our discussion in the present paper on the incentives of players to misreport their valuation to their own bidding agents is 
directly related to this meta-game model.

\section{Preliminaries}\label{sec:preliminaries}
We consider repeated auction settings in which the same group of bidding agents repeatedly play an auction game with a fixed auction rule and the same set of items are being sold at each repetition of the game. The valuation of every agent, assigned to it by its user, is constant throughout the dynamics, and the utilities are additive. I.e., a bidder with value $v$ for an item who wins $k$ such items has total value $kv$. Player values and bids are 
assumed to be on an $\epsilon$-grid, i.e., to be multiples of a minimum value $\epsilon > 0$ (e.g., $1$ cent).
We assume that agents cannot overbid, and in all auction formats if agents place equal bids, ties are broken uniformly at random. 

First-price and second-price auctions are defined in the standard way: the auction sells a single item, and $b_1 > b_2 > ... > b_n$ denotes the bids of $n$ bidders. Bidder $s\in [n]$ with the highest bid (bidding $b_1$) wins the auction, and pays a price $p$ to the auctioneer. In the first-price auction the winner pays his bid: $p=b_1$, and in the second-price auction the winner pays the amount of the next highest bid: $p = b_2$. 
The utility for player $s$ with value $v_s$ is then $v_s - p$ when winning the item, and zero otherwise.  
For further details on these auction formats see, e.g., \cite{nisan2007introduction}.  
Generalized auctions for multiple items are defined in the relevant sections. 

We study agents that use regret-minimization algorithms. The regret of player $i$ at time $T$, given a history of bids 
$(\textbf{b}^1, ... ,\textbf{b}^T)$, is defined as the difference between the optimal utility from using a fixed bid in hindsight and the actual utility: 
$R_i^T = \max_b \sum_{t=1}^T u_i(b, \textbf{b}_{-i}^t) - u_i(b_i^t, \textbf{b}_{-i}^t)$,  
where 
$b_i^t$ 
is the bid of player 
$i$ at time 
$t$ and 
$\textbf{b}_{-i}^t$ denotes the bids of the other players at time 
$t$. 
As usual, regret-minimization algorithms are stochastic, and whenever we talk about the limit behavior, we consider a sequence of algorithms\footnote{This follows most of the literature that looks at regret-minimizing algorithms with a fixed horizon $T$ and looks at a sequence of such algorithms as $T \rightarrow \infty$.
Typically, and specifically for the case of multiplicative weights, the ``update parameter'' depends on $T$
but remains fixed over all time steps $1 \le t \le T$ of the algorithm with horizon $T$.  While it is
alternatively possible to consider a single regret-minimizing algorithm that has an 
infinite horizon and uses a decreasing
update parameter \cite{auer2002adaptive,cesa2006prediction}, such a definition does not ``play well'' with the notion of 
mean-based algorithms \cite{braverman2018selling} 
that we use below, and so we do not use it.
}
with $T \rightarrow \infty$ and with probability approaching $1$. An agent $i$ is said to be ``regret-minimizing'' if $R_i^T/T \rightarrow 0$ almost surely as $T \rightarrow \infty$. In our simulations we mainly present the multiplicative-weights algorithm, as defined in \cite{arora2012multiplicative}. The results apply equally both to the Hedge algorithm (often also referred to as multiplicative weights in the literature) and to the linear version of the algorithm. Appendix \ref{sec:appendix-additional-simulations} presents simulations with additional algorithms. 

We use the following notations to describe the empirical dynamics of the agents' bids. Denote by $\Delta$ the space of probability distributions over bid tuples in an auction with discrete bid levels. We use $\textbf{p}^T_t\in\Delta$ to denote the empirical distribution of bids after $t$ rounds in a sequence of $T\geq t$ auctions, and denote by $p_t^T(\textbf{b})$ the empirical frequency of a bid tuple $\textbf{b}$ after round $t$. 
When we discuss convergence of the dynamics we consider the following definition: 

\begin{definition} The dynamics ``converge to a distribution'' $\textbf{p}\in\Delta$ if for every $\epsilon>0$ there exists $T_0(\epsilon)$ s.t. for every $T>T_0$ w.p. at least $1-\epsilon$ it holds that for every $\epsilon T<t\leq T$, $|\textbf{p}_t^T-\textbf{p}|<\epsilon$.
\end{definition}

Intuitively, this means that after a sufficiently long time, the time-average distribution stabilizes and remains close to a stationary distribution. 
Notice that the convergence of the cumulative empirical distribution to a joint distribution does not imply the convergence of the probability of play
of any single agent to a distribution on its own actions\footnote{
Notice also that convergence to a distribution $\textbf{p}$ does allow for the existence of an unlimited number of accumulation points (other than $\textbf{p}$),  
as long as all these points together have zero probability in the limit. However, the impact of such accumulation points on the average utilities and revenue is vanishing in the long run. 
} 
(a stronger notion of convergence studied, e.g., in \cite{bailey2018multiplicative,feng2020convergence,deng2021nash}).  See also \cite{KolumbusNisan2021manipulate} for a discussion of
yet weaker notions of convergence.

\section{Co-Undominated Coarse Equilibria}\label{sec:co-undominated-CCE}
In this section we define a refinement of the concept of coarse correlated equilibria, which we call ``co-undominated CCE.'' 
For simplicity, the definitions are given for two-player games; their generalization is straightforward. 
Intuitively, those are CCEs in which players do not use strategies that are weakly dominated relative to the distributions of the opponents. Formally: 

\begin{definition}
Let $A,B$ denote subsets of the action spaces of the two players in a two-player game. Action $i\in A$ of player $1$ is called \emph{weakly-dominated-in-$B$ 
by action $i'$} of player $1$ if: $\forall j \in B: \ u_1(i,j) \leq u_1(i',j)$, and: $\exists j \in B: \ u_1(i,j) < u_1(i',j)$.
\end{definition}

\begin{definition}
Let ${p_{i,j}}$ be a CCE of a (finite) two-player game with action spaces $I,J$.  Denote its support $(A,B)$ by $A =\{i\in I | \exists j\in J $ $ with \ p_{i,j}>0\}$  and $B=\{j\in J | \exists i\in I $ $ with \ p_{i,j}>0\}$ ($A$ is a subset of $I$ and $B$ is a subset of $J$). 
A CCE is called \emph{co-undominated} if no strategy in its support is weakly dominated relative to the support of the other player,  i.e., if for every $i \in A$, and every $i'\in I$, action $i$ is not weakly-dominated-in-$B$ by $i'$, and similarly for $B$.
\end{definition}

We show that these types of equilibria are closely connected to a family of learning algorithms called ``mean-based.'' 
For completeness, we cite their definition (\cite{braverman2018selling}, Definition $2.1$):  

\begin{definition}
(Mean-based learning algorithm \cite{braverman2018selling}). Let 
{\small$\sigma_{i,t} = \sum_{t'=1}^t u_{i,t'}$}, where {\small$u_{i,t}$} is the utility of action $i$ at time $t$. 
An algorithm for the experts problem or the multi-armed bandits problem is $\gamma(T)$-mean-based if it is the case that whenever
{\small$\sigma_{i,t} < \sigma_{j,t} - \gamma(T) T$}, then the probability that the algorithm pulls arm $i$ in round $t$ is at most $\gamma(T)$. We say an algorithm is mean-based if it is $\gamma(T)$-mean-based for some $\gamma(T) = o(1)$.
\end{definition}
The above definition from \cite{braverman2018selling} considers a fixed value $T$. To clarify the definition further when considering a series of $T$-values and the limit $T\rightarrow \infty$, we say that an algorithm family $A^T$, parametrized by $T$, is mean-based if there exists a series $(\gamma(T))_{T=1}^\infty$ such that $\gamma(T)\rightarrow 0$ as $T \rightarrow \infty$, and each algorithm in the series $(A^T)_{T=1}^\infty$ is $\gamma(T)$-mean-based.

The following lemma connects between mean-based regret-minimization algorithms and co-undominated CCEs. This result is general for any game (not only auctions), and will be useful for our analysis in the following sections.

\begin{lemma} \label{thm:mean-based-CCE-supports}
Consider the dynamics of any mean-based regret-minimizing agents playing a repeated finite two-player game. 
If the dynamics converge to a distribution $\textbf{p}$, then $\textbf{p}$ is 
a co-undominated CCE.
\end{lemma}
\begin{proof} 
Consider any finite two-player game played repeatedly by two mean-based regret-minimizing agents, and assume that the dynamics converge to a distribution $\textbf{p}\in \Delta$. First, as is well known, since the agents are minimizing regret, $\textbf{p}$ must be a CCE (otherwise some agent must remain with positive average regret in the limit). Next, assume by way of contradiction that action $i$ of player $1$ (w.l.o.g.) is weakly dominated by action $i'$ of player $1$ in the support of player $2$, and that $i$ is in the support of player $1$. Let us use $u(i,j)$ to denote the utility for player $1$ when player $1$ bids $i$ and player $2$ bids $j$, and $u_t(i)$ to denote the utility for player $1$ from using the bid $i$ at time $t$, given the actual (empirical) action of player $2$ at time $t$. Thus, for every action $j$ of player $2$ with $\Pr_{\textbf{b}\sim \textbf{p}}(b_2=j) > 0$ it holds that $u(i,j) \leq u(i',j)$, and there exists an action $j_0$ of player $2$ with $\Pr_{\textbf{b}\sim \textbf{p}}(b_2=j_0)=q_0 > 0$ such that $u(i,j_0) < u(i',j_0)$, and $\Pr_{\textbf{b}\sim \textbf{p}}(b_1=i)=p_0 > 0$. 

Denote $\delta = u(i',j_0) - u(i,j_0)$, and $\Delta_t^T = \sum_{s=1}^t u_s(i') - u_s(i)$.
Let $\epsilon > 0$ s.t. $\epsilon < \frac{q_{\scaleto{0}{3pt}} \delta}{8}$ and $\epsilon < \frac{p_{\scaleto{0}{3pt}}^2}{9}$, and take $T$ s.t. the algorithm used by player $1$ is $\gamma(T)$-mean-based, and $T>T_0(\epsilon)$, and $\gamma(T) < \epsilon^2$.

\vspace{5pt}
\noindent
\emph{Claim 1}: Fix any $0<t\leq T$. In the event that $|\textbf{p}_t^T - \textbf{p}|\leq \epsilon$ it holds that $\Delta_t^T/t > \epsilon$.

\vspace{3pt}
\noindent
\emph{Proof}: Notice that if $\textbf{p}_t^T = \textbf{p}$, then $\Delta_t^T = t q_0 \delta$. In the event that $|\textbf{p}_t^T - \textbf{p}|\leq \epsilon$, it holds that, specifically, $|\textbf{p}_t^T - \textbf{p}| \leq 2\epsilon$. To bound $\Delta_t^T$ from below, let us consider the case where $|\textbf{p}_t^T - \textbf{p}|= 2\epsilon$ and all the difference between the two distribution vectors in concentrated on action profiles in which the dominating action $i'$ has zero utility and the dominated action $i$ has a utility of $1$. Thus $\Delta_t^T > (1-\epsilon)q_0 \delta t - \epsilon t$, 
and therefore $\frac{\Delta_t^T}{t}> q_0 \delta - 2 \epsilon \geq \frac{3}{4}\epsilon q_0 \delta > \epsilon$. 

\vspace{5pt}
\noindent
\emph{Claim 2}: Denote $p_t^T(i) = \sum_j p_t^T(i,j)$. For every $t>T_1 = \sqrt{\epsilon} \cdot T$ 
with probability more than $1-2\epsilon$, it holds that $p_t^T(i) < 2 \sqrt{\epsilon}$.
 
\vspace{3pt}
\noindent
\emph{Proof}: By claim $1$ above, for every $t>\epsilon T$, in the event that $|\textbf{p}_t^T - \textbf{p}|\leq \epsilon$ it holds that $\Delta_t^T > \epsilon t \geq \epsilon^2 T$, and thus from the mean-based property of player $1$, it follows that the probability that action $i$ is chosen at time $t+1$ is less than $\gamma(T)<\epsilon^2$.
From the convergence assumption, this holds with probability at least $1-\epsilon$ for any $t>\epsilon T$. 
Thus, with high probability (at least $1-\epsilon$) the empirical average number of times that action $i$ is chosen in rounds $\epsilon T, ... , t$ is not far from its expectation of $\epsilon^2$, which is less than $\epsilon$.
To bound $p_t^T(i)$ from above, consider the ``worst case'' where $p_{\epsilon T}^T(i) = 1$ (until time $\epsilon T$). 
Thus, for $t>T_1$, with probability more than $1 - 2 \epsilon$ it holds that  
$p_t^T(i) \leq  
\frac{1}{t}(\epsilon T + (t-\epsilon T)\epsilon) < 
\frac{1}{t}(\epsilon T + \epsilon t)
\leq \left(1 + \frac{T}{T_1}\right)\epsilon = 
\epsilon + \sqrt{\epsilon} < 2 \sqrt{\epsilon}$.

\vspace{3pt}
We thus have by claim $2$ that with probability more than $1-2\epsilon$, $p_T^T(i) < 2 \sqrt{\epsilon}$, and since $p_0 - 2\sqrt{\epsilon} > 3 \sqrt{\epsilon} - 2 \sqrt{\epsilon} = \sqrt{\epsilon} > \epsilon$, we have that $\Pr \left(|\textbf{p}_T^T - \textbf{p}|> \epsilon\right) > 1 - 2 \epsilon$, in contradiction to the convergence assumption. Therefore, the dominated action $i$ cannot be in the support of player $1$. This holds for the other player as well, and thus the distribution $\textbf{p}$ must be co-undominated.
\end{proof}

\noindent
Remark:   
Notice 
that this argument can be extended also to the case where the dynamics converge to $\textbf{p}$ with a strictly positive probability that may be 
less than $1$.

\section{Second-Price and GSP Auctions}\label{sec:second-price}
We consider agents that play repeated second-price auctions (for a single item) and GSP auctions (for multiple items; see, e.g., \cite{edelman2007internet}). 
Our analysis starts with the simple second-price auction, which is known to be dominant-strategy incentive compatible; i.e., it is a dominant strategy for all players to bid their true values. We show  that in the setting where regret-minimizing agents play the auction on behalf of their users, the outcomes are different from this truthful equilibrium. 

Specifically, Theorem \ref{thm:second-price-auction-MW-bid-distributions} shows that when the auction is played between two multiplicative-weights agents, 
in the limit, although the high player always wins the auction (as in the truthful equilibrium), he pays a price that is strictly less than the second price. 
Before proceeding to the proof, we consider the following technical restatement of the theorem.  

\begin{theorem*} (Restatement of Theorem \ref{thm:second-price-auction-MW-bid-distributions}): 
In the repeated second-price auction with discrete bid levels that are multiples of $\epsilon$ and players using multiplicative-weights algorithms with values $v > w$, in the limit the dynamics converge to a joint distribution in which for the ``high player'' with value $v$, $Pr[0]=Pr[\epsilon]=...=Pr[w] = 0$, and $Pr[w+\epsilon]=Pr[w+2\epsilon]=...=Pr[v] = \frac{\epsilon}{v-w}$, and for the ``low player'' with value $w$, $0 < Pr[0] \leq Pr[\epsilon] \leq ... \leq Pr[w]$.  
\end{theorem*}

\begin{proof} 
First, observe that for the low player, bidding $w$ weakly dominates all other bids, and so the probability assigned to this bid does not decrease in any update of the weights. Thus, the probability that the low player bids exactly $w$ is at least\footnote{For simplicity, we use the standard definition of the algorithm where the weights are initialized uniformly. The result that the high player always wins and pays strictly less than the second price extends to any other initialization that has positive probability for every bid.} 
$1/(w/\epsilon+1)$.

Since the low player bids $w$ at least $1/(w/\epsilon+1)$ of the time (with high probability), every bid that is less than $w$ of player $1$ loses the auction at this frequency, while an alternative bid of $w+\epsilon$ always wins, yielding positive utility. Thus, every $O(w/\eta \epsilon)$ steps the probability that the high player bids $w$ or less decreases at least by a constant factor. 

Hence, we can deduce that the expected number of times that the high player bids below $w$ is a constant, since it is bounded by the sum of a geometric series. 
In light of this observation, we can look at the bid distribution of the low player. Whenever the high player bids above $w$, the weights of all bids do not change (since all bids would yield zero utility). The weights are updated only in those cases where the high player bids below $w$. Since this happens only a finite number of times, the ratio of probabilities between any two bids of the low player remains bounded by a constant.
Thus, every bid of the low player remains with positive probability in the limit (and at any time). 
Since whenever a bid $j$ of the low player would win the auction, every bid $k>j$ would also win the auction, it follows that $Pr[0] \leq Pr[\epsilon] \leq ... \leq Pr[w]$.

Finally, since all bids above $w$ have the same utility at every step (as the payment is the second price), all weight updates of every bid above $w$ are exactly equal, and so the distribution of bids of the high player in the limit is uniform between $w+\epsilon$ and $v$. 
\end{proof}

In addition to the above result that explicitly shows that dynamics of multiplicative-weights agents converge to these forms of distributions, 
our analysis from the previous section allows us to show that for two agents implemented as any mean-based regret minimization algorithms, whenever the dynamics of the agents converge to a coarse correlated equilibrium, it must be an equilibrium in which the agent with the higher valuation always wins.

\begin{lemma}
In the repeated single-item second-price auction with discrete bid levels that are multiples of $\epsilon$ and two players with values $v > w$, in any   co-undominated CCE the player with the higher valuation $v$ always wins the auction.
\end{lemma}
\begin{proof} Denote the high player with value $v$ as player $1$ and the low player with value $w$ as player $2$. Assume that the players are in a CCE in which there is positive probability that player $1$ loses the auction. Thus, there exist $i_0,j_0$ such that $0 = u_1(i_0,j_0) < u_1(w+\epsilon, j_0) = v-j_0$.
If for all $j \in B$ ($B$ being the support of player $2$) it holds that $u_1(i_0,j) \leq u_1(w+\epsilon,j)$, then $w+\epsilon$ weakly dominates $i_0$, leading to a contradiction with the co-undominated CCE condition. If we assume the contrary, we have that there exists $j\in B$ such that $u_1(i_0,j) > u_1(w+\epsilon,j)$. This is impossible since bidding $w+\epsilon$ wins the auction and gives utility $v-j$, while any other bid gives either the same utility or zero. Therefore, in any co-undominated CCE the high player always wins.  
\end{proof}

These theoretical results are consistent with simulation results, as can be seen in Figure \ref{fig:second-price-suctions} that shows    
the bid dynamics of two multiplicative-weights agents playing a repeated second-price auction for a single item. The first player has value $1$ for the item and the second player has value $0.5$. In the very beginning, bids are close to uniform and there are some auctions in which player $2$ bids higher than player $1$ and wins the item with positive utility. However, after a short learning phase player $1$ learns to bid higher than the value $0.5$ of player $2$ and to always win the auction. 

\begin{figure}[!t]
\centering
\vspace{-24pt}

	\begin{subfigure}{.49\linewidth}
		\includegraphics[width=1.0\linewidth]{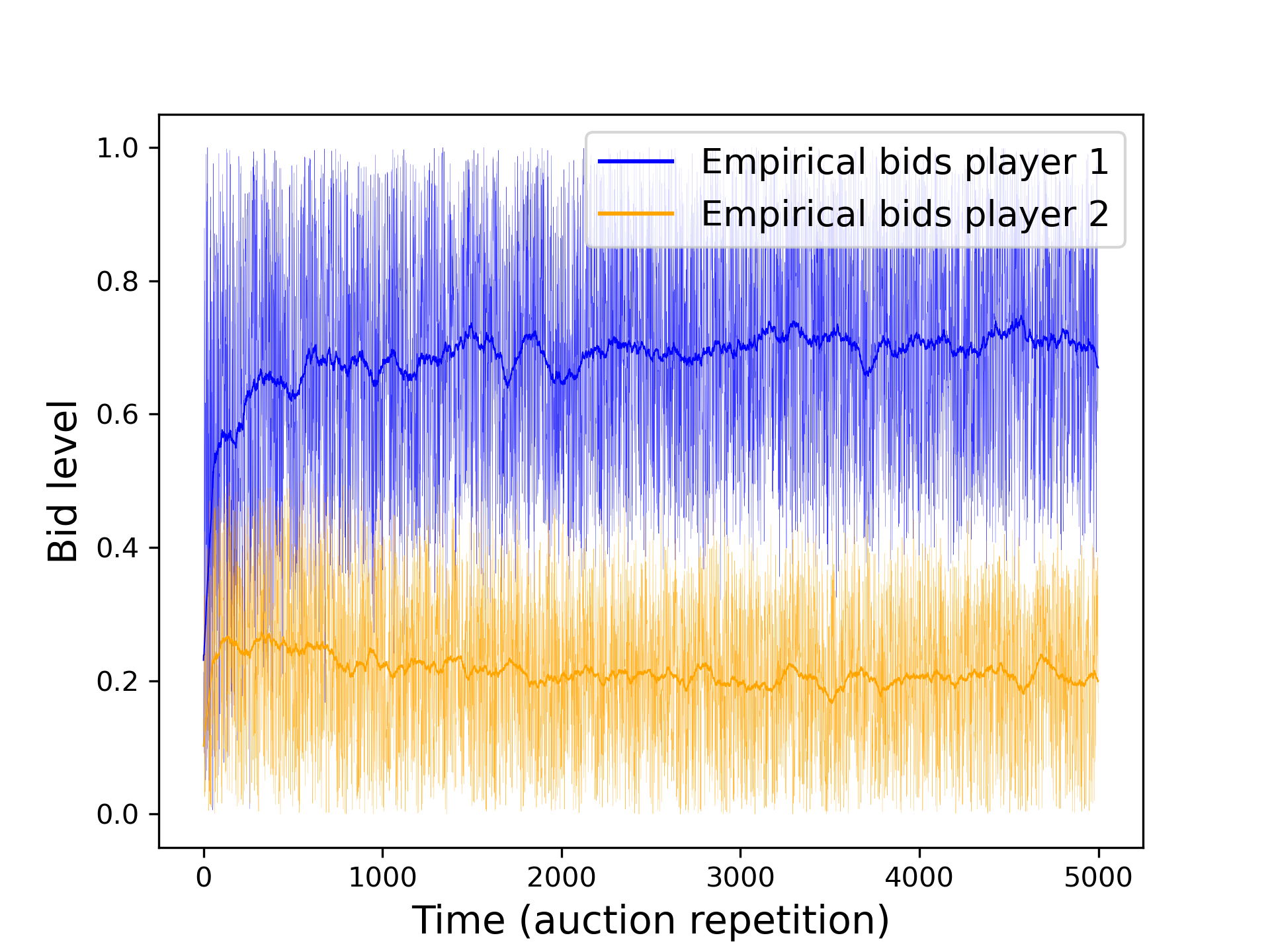}
		\caption{Bid dynamics}
		\label{fig:GSP_v=01_w=05_MW_bid_dynamics}
	\end{subfigure}
	\begin{subfigure}{.49\linewidth} 
		\includegraphics[width=1.0\linewidth]{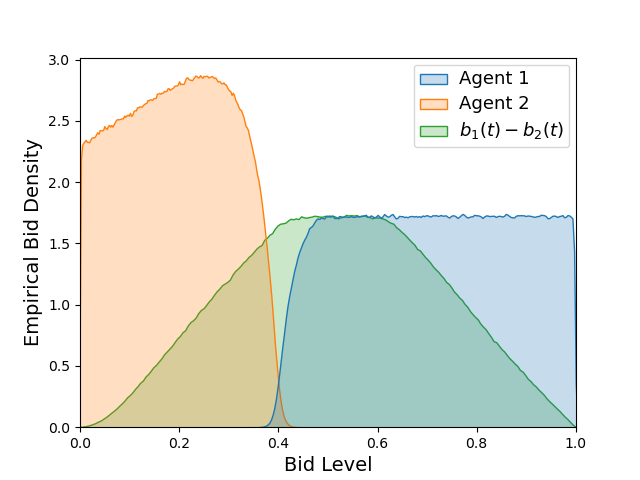}
		
		\caption{Bid distributions}
		\label{fig:GSP_BidPDF}
	\end{subfigure}
	\vspace{-2pt}

\caption{{\small GSP auctions. 
Figure \ref{fig:GSP_v=01_w=05_MW_bid_dynamics}: Bid dynamics in a GSP auction with two items with CTRs $1$ and $0.5$ and player values $v=1$ and $w=0.5$. 
Figure \ref{fig:GSP_BidPDF}: Long-term bid distribution in the same auction.
}}
	\vspace{-5pt}

\label{fig:SP_gains_from_manipulation_and_GSP_auctions}
\end{figure}

Figure \ref{fig:SecondPriceV1W05BidDensity} shows the empirical bid densities of the two players. 
It can be seen that the high player bids uniformly between his value and the value of the low player. The basic reason for this is that under the second-price payment rule, all these bids are equivalent. The bids of player $2$ are not uniform, but have more weight on higher bids. This is despite the fact that in the long-term dynamics player $2$ always loses, getting zero utility, and so all his bids should also be equivalent. The reason for this form of the distribution is that the multiplicative-weights algorithm (and other algorithms like FTPL) keep a ``memory'' of the history of play: in the initial phase there was a finite number of auctions in which player $2$ got positive utility. These winnings were associated with higher bids. 
Thus, in the initial phase higher bids get higher weight, and after this phase all bids yield exactly the same feedback (of zero utility), and therefore the distribution remains skewed toward the higher bids. This intuition is formalized in the proof of Theorem \ref{thm:second-price-auction-MW-bid-distributions} above.  

Figure \ref{fig:GSP_v=01_w=05_MW_bid_dynamics} shows a very similar picture also in the generalized second-price auction with multiple items. Interestingly, there is a narrow bid range (centered at the bid $0.4$ in Figure \ref{fig:GSP_BidPDF}) in which the bid supports of both bidders overlap, but, as can be seen by the distribution of bid differences at all times (in green), the bids are correlated such that despite the overlapping bid ranges, the player with the higher valuation still almost always wins the auction.  

Now let us consider the point of view of the human users that report their values to their regret-minimizing agents playing a repeated second-price auction on their behalf. Although this auction is known to be incentive compatible, in light of the results shown above on the dynamics of the agents, from the point of view of the human users the second-price outcomes are not obtained, and thus these users do not observe a second-price auction. Hence, in the repeated second-price auction that is implemented with regret-minimizing agents it can be beneficial for players to misreport their true valuations.  

To demonstrate this, let us return to our example from the introduction in which two players compete in a repeated second-price auction using multiplicative-weights agents. The first player has value $0.4$ and the second player has value $0.5$. Figure \ref{fig:SP_gains_from_manipulation} shows the expected utility for player $1$ over a period of $50$,$000$ auctions from inputting different values $v$ into his agent. It can be seen that when $v$ is sufficiently high above the value $0.5$ of player $2$, the utility of player $1$ is positive and increases for yet higher values; i.e., for a fixed value of the low player, the high player pays a lower price as he declares a higher value. For very high values of the high player, the high player's agent almost never bids below the low value, and so the distribution of bids of the low player remains very close to uniform.  
As the high declaration decreases, the probability that the high player's agent bids below the low value increases, and this is what drives the increase in the probability of the low player's agent bidding higher values.  

It follows that the best reply of the high player in the meta-game is always as high a declaration as possible, as long as the low player's declaration is below a maximum value that is less than twice the true value of the high player, and under which the dynamics lead to a positive average utility for the high player.    
The low player (almost) always loses and so all declarations that cause him to lose are equivalent.  
These can lie anywhere between $0$ and the declaration of the high player.   
The truthful declaration is certainly one possible best reply in the full information game, 
but notice that there is little to suggest selecting it specifically since it is a best reply only if it leads to a loss, and is in no way a dominant strategy if the value of the other player is not known. 
We thus can identify a whole set of possible equilibria of the meta-game, but do not see a natural way to select one of these as the outcome of user behavior in the meta-game. Notice that the revenue is close to half of the low declared value (when the high declaration is very high) 
and thus we cannot give a reasonable prediction of the revenue in equilibrium. 

\section{First-Price Auctions}

In this section we present our analysis of first-price auctions leading to Theorem \ref{thm:first-price-auctions}. 
We begin by showing that in the first-price auction in any co-undominated CCE, all bids of the high-value player are close to the second price. This is formalized in the following lemma.  

\begin{lemma}\label{thm:FP-CCE-with-unduminated-supports}
For the first-price auction with discrete bid levels that are multiples of $\epsilon$ with player values $v\geq w>2\epsilon$, in any co-undominated CCE, 
the support of the ``high player'' with value $v$ is in $[w-2\epsilon,...,w+\epsilon]$. Additionally, when $\epsilon$ is small compared to the difference in values $v-w$, the player with the higher value $v$ almost always wins the auction. 
\end{lemma}

To prove Lemma \ref{thm:FP-CCE-with-unduminated-supports}, we start by analyzing the symmetric case of bidders with equal values, which is technically slightly different. The proof of the next lemma is deferred to Appendix \ref{sec:appendix-first-price}. 

\begin{lemma} \label{thm:symmetric-FP-CCE-supports}
For the single-item first-price auction with equal integer values $v>2\epsilon$ and bid levels that are multiples of $\epsilon$, any co-undominated CCE is supported on $[v-2\epsilon,...,v]\times[v-2\epsilon,...,v]$.
\end{lemma}

We now proceed to prove Lemma \ref{thm:FP-CCE-with-unduminated-supports}.

\begin{proof} (Lemma \ref{thm:FP-CCE-with-unduminated-supports}): 
Denote the value of player $1$ as $v$ and the value of player $2$ as $w$ and assume w.l.o.g. that $v \geq w$. 
Lemma \ref{thm:symmetric-FP-CCE-supports} deals with the symmetric case $v=w>2\epsilon$.  
Here we will show that when $v > w > 2\epsilon$, bids that are not in $[w-2\epsilon,...,w+\epsilon]$ cannot be in the support of player $1$. 

First, for any support $B$ of player $2$, any bid $b>w+\epsilon$ is strictly dominated for player $1$ by bidding $w+\epsilon$. Thus, in any co-undominated CCE, the high player does not bid above $w+\epsilon$. 

Next, let $b_1, b_2$ be the lowest bids of the players in the support and let $B_1,B_2$ be the highest bids in the support, and assume that $b_1 < w-2\epsilon$. 
We will consider the following cases: 

\vspace{3pt}
\noindent
Case I: $b_1=B_1$.  Here player $1$ plays a pure strategy and so the CCE condition implies that player $2$ must best-reply to it. 
Thus, for $b_1=B_1 < w-2\epsilon$, we must have $b_2=B_2=b_1+\epsilon$,
which is not an equilibrium, since player $1$ would then always lose and thus prefers using a fixed bid $w-2\epsilon$. 
Therefore, it must be that $b_1 < B_1$.

\vspace{3pt}
\noindent
Case II: $b_1=b_2$. Here, if $b_1=b_2<w-2\epsilon$ then $b_1$ is strictly dominated by $b_1+\epsilon$ ($b_1$ loses unless player $2$ bids $b_2$, in which case player $1$'s utility is $(v-b_1)/2$ due to the tie, but then $v-(b_1+\epsilon)$, which this player gets from bidding $b_1+\epsilon$, is better).  Thus, we have $b_1=b_2 \geq w-2\epsilon$.  

\vspace{3pt}
\noindent
Case III: $b_1<b_2$. In this case $b_1$ always loses and always gets a utility of $0$. Thus, due to the co-undominated CCE condition, player $1$ must get zero utility for every bid in his support and every bid in player $2$'s support, and so he gets zero utility. This leads to a contradiction since player $1$ can bid $w$ and have positive utility in every auction. 

\vspace{3pt}
\noindent
Case IV: $b_1>b_2$. Here $b_2$ always loses and always gets zero utility. Thus, due to the co-undominated CCE condition, player $2$ must get zero utility for every bid in the support. However, player $2$ can bid $w-2\epsilon$ and have a positive utility of $2\epsilon$ in every auction in which player $1$ bids $b_1$, a contradiction. \vspace{1pt}

Hence, the support of the high player must be in $[w-2\epsilon,...,w+\epsilon]$. 
Next, let us show that the high player (with value $v$) wins the auction with probability approaching $1$ when the grid step  $\epsilon$ is small compared to the differences in values. Assume that there is positive probability $p$ that player $1$ bids $a$ and player $2$ bids $b$ and $a\leq b$. 
We know that $w-2\epsilon \leq b \leq w$, and $w-2\epsilon \leq a \leq w+\epsilon$, and denote $\delta=v-w > 2\epsilon$. On the one hand, the expected utility for player $1$ from bidding $a$ in a sequence of $T$ auctions is at most $T\cdot (1-p)(v-w+\epsilon)$. On the other hand, the utility for player $1$ from placing a fixed bid of $w+\epsilon$, which always wins the auction, is $T\cdot (v-w-\epsilon)$. Thus, from the CCE condition it follows that $p \leq \frac{2\epsilon}{v-w+\epsilon}$, and therefore $p$ approaches zero as $\epsilon / \delta$ approaches zero, i.e., as the grid becomes small compared to the difference in values. 
\end{proof}
Using this result together with Lemma \ref{thm:mean-based-CCE-supports} that connects between mean-based learning agents and co-undominated CCEs, we can readily prove Theorem \ref{thm:first-price-auctions}.  
\begin{proof} (Theorem \ref{thm:first-price-auctions}): 
Consider a first-price auction with discrete bid levels that are multiples of $\epsilon$ with players with values $v\geq w>2\epsilon$ that repeatedly play the auction using mean-based regret minimization algorithms. By Lemma \ref{thm:mean-based-CCE-supports}, if the dynamics converge to a CCE, then it must be a co-undominated CCE. By Lemma \ref{thm:FP-CCE-with-unduminated-supports}, in the first price auction in such an equilibrium both players bid close to the low value $w$ (up to $2\epsilon$), and the player with the higher value $v$ wins with probability approaching $1$ as $\epsilon/(v-w) \rightarrow 0$.
\end{proof}

\section{Generalized First-Price Auctions}\label{sec:GFP}
We analyze GFP auctions where two ``ad slots'' are being sold 
 to two bidders. The ``top'' ad slot has a ``click-through rate'' of $1$, while the ``bottom'' ad slot  
has a click-through rate of $0.5$. 
That is, a bidder with value $v$ that wins the top slot and pays $p$ (per click) gets a utility of $v-p$, 
and the bidder that gets the 
bottom slot and pays $p$ (per click) gets a utility of $v/2-p/2$.  The auction uses the first-price rule:
the bidder that bids highest gets 
the top slot and the low bidder gets the bottom slot, and each pays their bid per click.  
That is, if the first player has value $v$ and bids $x$ 
while the second player has value $w$ and bids $y$, then the utility of the first player is $v-x$ if $x>y$ 
and is $(v-x)/2$ if $x<y$, while 
the utility of the second player is $w-y$ if $y>x$ and is $(w-y)/2$ if $y<x$.  
(In the case of a tie, $x=y$, let us assume that we toss a coin for the top slot, in which case the utility of the first bidder is $3\cdot(v-x)/4$, and the utility of the second bidder is $3 \cdot (w-y) /4$.) 
We start with the proof of Theorem \ref{thm:GFP-mixed-Nash} on the Nash equilibrium of the GFP auction, and then proceed to analyze the empirically observed distributions obtained by the agents.

\subsection{Nash Equilibrium}

\begin{proof} (Theorem \ref{thm:GFP-mixed-Nash}):
Denote the bid density function of player $1$ in a Nash equilibrium by $f(x)$ and the bid density function of player $2$ in a Nash equilibrium by $g(y)$, and denote the respective cumulative density functions by $F(x)$ and $G(y)$.
First, observe that every bid $y > w/2$ of the second player is dominated by $y=0$. Therefore, the support of the distributions $f(x),\ g(y)$ is in\footnote{Since $y=w/2$ is only weakly dominated by $y=0$, the bid $y=w/2$ may have positive probability in equilibrium. In principle, if $g(y)$ in equilibrium had a finite (atomic) weight on the bid $y=w/2$, then the support of $f(x)$ would need to be extended to $[0,w/2+\varepsilon]$ for some $\varepsilon > 0$. However, our solution for $g(y)$ shows that the bid $y=w/2$ has zero probability and thus the support $[0, w/2]$ is sufficient for $f(x)$.}  $[0, w/2]$. 
Next, the expected utility to player $1$ from bid $x$, assuming that player $2$ plays according to $g(y)$, is    
\small
$
u_1(x) = \int_{0}^{x} g(y)(v-x) \,dy + \int_{x}^{w/2} g(y)\frac{1}{2}(v-x) \,dy = \frac{1}{2}(v-x)(G(x)+1)
$, 
\normalsize
where the first term is the expected payoff from winning the first slot in the auction, and the second term is the expected payoff from getting the second slot. The expected utility to player $2$ from bid $y$, assuming that player $1$ plays according to $f(x)$, is 
\small
$
u_2(y) = \int_{0}^{y} f(x)(w-y) \,dx + \int_{y}^{\frac{w}{2}} f(x)\frac{1}{2}(w-y) \,dx = \frac{1}{2}(w-y)(F(y)+1)
$. 
\normalsize
In equilibrium the derivatives of the utilities must satisfy 
\small
\begin{equation*}
\frac{d u_1(x)}{dx} = (v-x)g(x) - G(x) - 1 = 0,  \ \frac{d u_2(y)}{dy} = (w-y)f(y) - F(y) - 1 = 0.
\end{equation*}
\normalsize

With another derivative, we obtain differential equations for $f$ and $g$. 
We also replace the variable names by their original notations:
\small
\begin{equation*}
(v-y)g'(y) - 2g(y) = 0,  \quad   (w-x)f'(x) - 2f(x) = 0.
\end{equation*}
\normalsize

The solution to these equations is of the form {\small$f(x) = \frac{c_1}{(w-x)^2}$, $g(y) = \frac{c_2}{(v-y)^2}$, where $c_1, c_2$} are constants that depend on the boundary conditions. The normalization constraint yields the unique solution  
\small
$
f(x) = \frac{w}{(w-x)^2}, \quad  g(y) = \left(2\frac{v}{w} -1 \right)\frac{v}{(v-y)^2}
$.
\normalsize
The integrals of these density functions lead to the cumulative densities stated in the theorem.

The utilities to the two players, with values $v$ to player 1 and $w$ to player 2 such that $v \geq w$, in the Nash equilibrium are then given by the following integrals. 
\small
$
u_1^{NE} = \int_0^{w/2} f(x)\Big(\frac{1}{2}(v-x)(G(x)+1)\Big) dx =\frac{v}{2} + (v - w)(1 - ln(2))
$
%
, 
$
u_2^{NE} = \int_0^{w/2} g(y)\Big(\frac{1}{2}(w-y)(F(y)+1) \Big) dy = \frac{w}{2}
$.
\normalsize
\end{proof}

\subsection{Coarse Correlated Equilibria}

As discussed in the introduction, the dynamics of multiplicative-weights agents with equal values $v=w=1$ yield, on average, the Nash equilibrium distribution as stated in Theorem \ref{thm:GFP-mixed-Nash}; however, unlike in the mixed equilibrium, here bids are highly correlated. 
Throughout the dynamics, as shown in Figure \ref{fig:GFP_symmetric_bid_dynamics}, both agents bid nearly the
same, and the bids do not converge. The lack of convergence to a Nash equilibrium is not surprising since we expect to reach only a CCE of the game, and not necessarily the Nash equilibrium. So let us analyze the empirically selected CCE.   

The fact that the empirical bids are all concentrated on the diagonal may be understood by considering best-reply dynamics, which may be viewed as a primitive ancestor of our no-regret dynamics. The best reply to a bid $x$ is always either $x+\epsilon$, or 0 if $x$ is more than half the player's value. Thus, in such dynamics we will tend to see bids that progress like $(x, x+\epsilon, x+2\epsilon, x+3\epsilon,...)$ until $x$ reaches $w/2$, at which point the second player will resort to bidding $0$, and then the bids will rise gradually again. This is indeed exactly the dynamics that we see in Figure \ref{fig:GFP_symmetric_bid_dynamics}. Looking at the empirical distribution of $x$ when $v=w=1$, it seems to be identical to the distribution 
{\small $F(x) = \frac{x}{1-x}$} derived for the Nash equilibrium. Indeed, we prove that the diagonal distribution $(x,x)$ where $x$ is thus distributed is a CCE of the auction with $v=w=1$. The formal statement of this result is deferred to Appendix \ref{sec:appendix-GFP}. 
We should note, though, that the above {\small$F$} is {\em not} the only distribution yielding  
a CCE and that we cannot theoretically justify why this particular one is the empirical distribution observed.  
One should also note that the utility obtained by the two agents in this CCE is $3/4\cdot \ln(2) \cong 0.52$, 
which is (slightly) higher than the utility of $0.5$ that two players would obtain in the Nash equilibrium.  
We might view this as a weak form of implicit collusion between the agents
in the auction.

\begin{figure}[!t]
\centering
\vspace{-36pt}
	\begin{subfigure}{.301\linewidth}
		\includegraphics[width=1.02\linewidth]{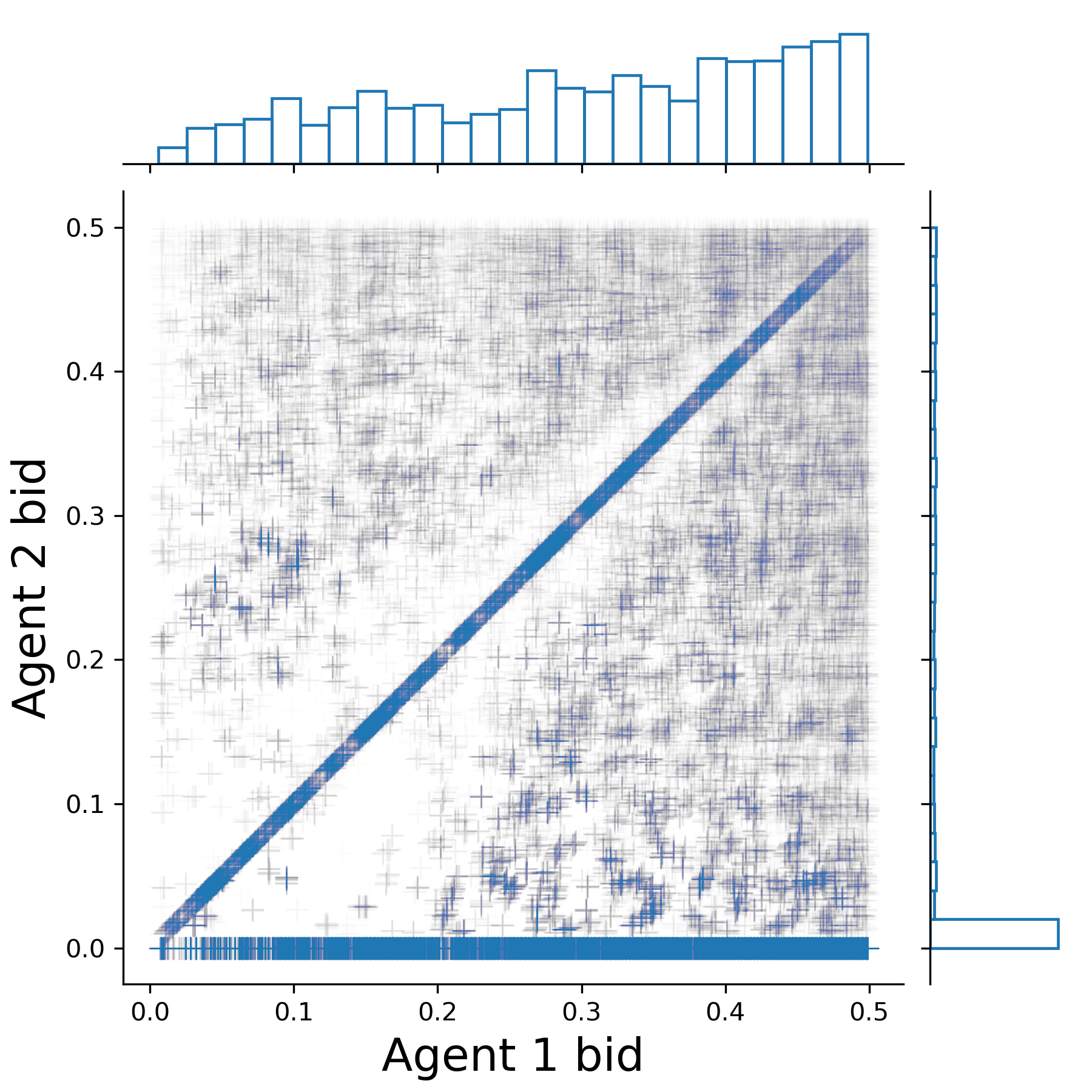}
		\caption{Joint distribution of bids of agents with $v=2$ and $w=1$.}
		\label{fig:GFP_joint_bid_dist}
	\end{subfigure}
	\begin{subfigure}{.341\linewidth} 
	\vspace{13pt}
		\includegraphics[width=1.1\linewidth]{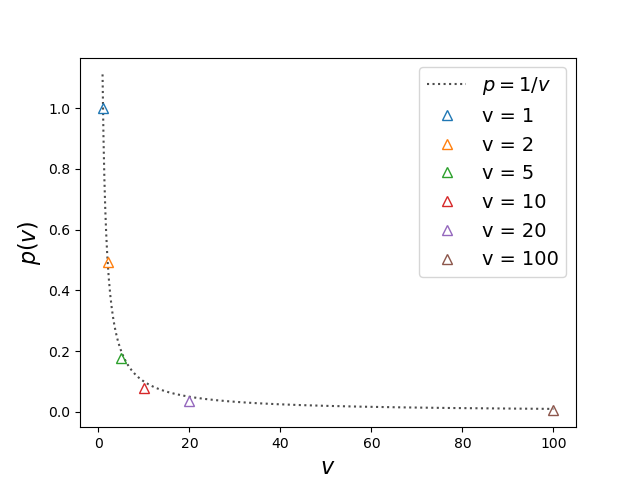}
		
		\caption{Probability of non-zero bids by the low agent vs. $v$ of the high agent.}
		\label{fig:p_v}
	\end{subfigure}
		\begin{subfigure}{.341\linewidth}
			\vspace{13pt}
		\includegraphics[width=1.1\linewidth]{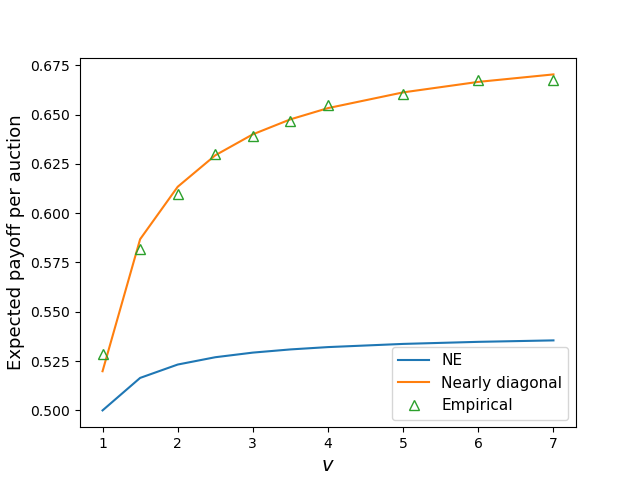}
		\caption{Gains from manipulation by the high-agent user with a true value of $1$.}
		\label{fig:GFP_unilateral_manipulation}
	\end{subfigure}
\caption{{\small Generalized first-price auctions played by multiplicative-weights agents with high value $v \geq 1$ (for agent $1$) and low value $w=1$ (for agent $2$).}}
\label{fig:asymmetric-GFP}
\vspace{-5pt}
\end{figure}

If one looks at no-regret dynamics when $v \ne w$, then the empirical distribution obtained is a bit more complex: not only
do we see bids on the diagonal, but we also see bids where the low agent bids $0$. This is shown in Figure \ref{fig:GFP_joint_bid_dist} for the case $v=2, w=1$. 
That is, in most rounds of the repeated auction either the two agents place nearly equal bids or the low agent bids zero.
We call such distributions ``nearly diagonal.'' 
Empirically we see that the distribution on $x$ (the bid of the agent with high value $v$) remains similar to 
the ``Nash equilibrium-derived'' distribution {\small$F(x) = \frac{x}{w-x}$}, while the probability of bidding on the diagonal, i.e.,
of the low agent bidding $y=x$ rather than $y=0$, decreases like $w/v$ as $v$ increases relative to $w$.  
Figure \ref{fig:p_v} depicts the empirical fraction of non-zero bids placed by the low 
agent as a function of the declared value $v$ of agent $1$ when $w=1$; the dotted line in the figure shows the comparison to the function $1/v$.

We thus consider the following to be our ``prediction'' of the outcome of the agent dynamics. 
First a value $x$ is chosen according to the cumulative distribution {\small$F(x) = \frac{x}{w-x}$} and 
then the high agent (the one with value $v$, where $v \ge w$) bids $x$ and the low agent bids
$x$ with probability $w/v$ and bids $0$ otherwise. We emphasize that we do not have a theoretical justification for this prediction and leave the rigorous derivation of the distribution selected by the agent dynamics as an open problem. 
Figure \ref{fig:GFP_unilateral_manipulation} compares the empirical utilities of the players 
to those implied by the Nash equilibrium and to those implied by our prediction. As we see, 
while the observed empirical utilities  are very far from those in the Nash equilibrium, 
they agree with our prediction quite well.

\subsection{Equilibria of the Meta-game}
We now proceed to study the ``meta-game'' faced by the users, each of whom needs to report his value to his own multiplicative-weights regret-minimization algorithm. Our analysis concerns the case where the true valuations of the players are $1$ and $1$ and their declarations to their agents are $v$ and $w$.  
Assuming that the agents reach the ``equilibrium'' predicted by our nearly diagonal distribution model, 
we can calculate the players' utilities. Our analysis allows us then to identify the ($\epsilon$-)Nash equilibria 
of the meta-game. The formal details of this analysis are  deferred to Appendix \ref{sec:appendix-GFP}. 
In the equilibrium of the meta-game, the high player declares an arbitrarily high value $v$ and the low player declares a fixed
value $w=1/(6 \cdot (1-\ln(2)) \cong 0.54$, and the utilities that they get, respectively, 
in this equilibrium are $5/6-O(1/v)$ and $1/2+O(1/v)$.  (If there exists a maximum allowed value for $v$ this is exactly 
an equilibrium, and otherwise it is an $\epsilon$-equilibrium, where $\epsilon=O(1/v)$.)
The interesting point of this analysis is that the two players are locked in a type of hawk-dove game: 
there are two possible equilibrium points depending on which player is the high bidder and which is the low one.  
Each player prefers to be the high player and to let the other player be the low player.  However, if they both 
declare high values, then they both suffer and get negative utility. 

This theoretically predicted behavior is indeed borne out in simulations.
Figure \ref{fig:w_manipulation_fixed_v=2} shows the average payoffs of the two players when player $1$ uses a fixed ``high'' declaration $v=2$  
as a function of the declaration $w$ of the other player. Figure \ref{fig:v_manipulation_fixed_w=054} shows the payoffs when player $2$ uses the equilibrium declaration $w=1/(6 \cdot (1-\ln(2)))$ as a function of the declaration $v$ of player $1$. It can be seen that the nearly diagonal distribution model is consistent with the empirical payoffs, and that the predicted equilibrium of the meta-game is indeed empirically an equilibrium as the best replies of the players behave as predicted: the best-reply of the ``low player'' is obtained at $w\cong 0.54$ (marked in red in Figure \ref{fig:w_manipulation_fixed_v=2}), 
and the payoff of the ``high player'' is increasing in the declaration $v$, approaching $5/6 \cong 0.83$ for high declaration values 
(Figure \ref{fig:v_manipulation_fixed_w=054}). These results are robust for different player parameters and several other learning algorithms such as linear or exponential multiplicative weights, FTPL, and FTPL with a recency bias.   

To summarize, we succeeded to identify (and empirically validate) the outcome of dynamics of 
multiplicative-weights agents in this auction with any two values. This allowed us to analyze the equilibrium in the meta-game induced on users that have identical true values. We found two such equilibria where, in each of them, one player takes the role of the high player, declaring the maximum possible value to his agent, and the other player declares a value of $0.54$ fraction of the true value.  Significantly, in each of these equilibrium points
the users capture $8/9 \cong 89\%$ of the welfare, 
leaving only $11\%$ to the auctioneer.  This outcome exhibits significant implied collusion between the users and, in fact, 
suggests that such an auction platform encourages collusion, probably against the interests of the platform designer.

\vspace{5pt}
\subsection*{Acknowledgments}
This project has received funding from the European Research Council (ERC) under the European Union's Horizon 2020 Research and Innovation Programme (grant agreement no. 740282).
\vspace{5pt}
\bibliography{no_regret_auctions_WWW2022_refs}
\appendix 
\newpage
\section*{Appendix}
\section{Additional Simulations}\label{sec:appendix-additional-simulations}
Our simulations show similar convergence results for different variants of multiplicative weights and FTPL algorithms. The following figures show the bid dynamics in second-price and first-price auctions when the first agent has a value of $v=1$ and the second agent has a value of $0.5$, for the Hedge algorithm and FTPL. We also look at simulations with a variant of FTPL with a recency bias that uses a geometric discount on the (perturbed) history of play. In all these different dynamics, in the second-price auction the high player always wins the auction and pays a price that is, on average, well below the actual second price of $0.5$, and in the first-price auction the high player always wins and pays the second price. 

\subsection{Second-Price Auctions}\label{sec:appendix-second-price}

\begin{figure}[!h]
\centering
\vspace{-14pt}
	\begin{subfigure}{.32\linewidth}
		\includegraphics[width=1.05\linewidth]{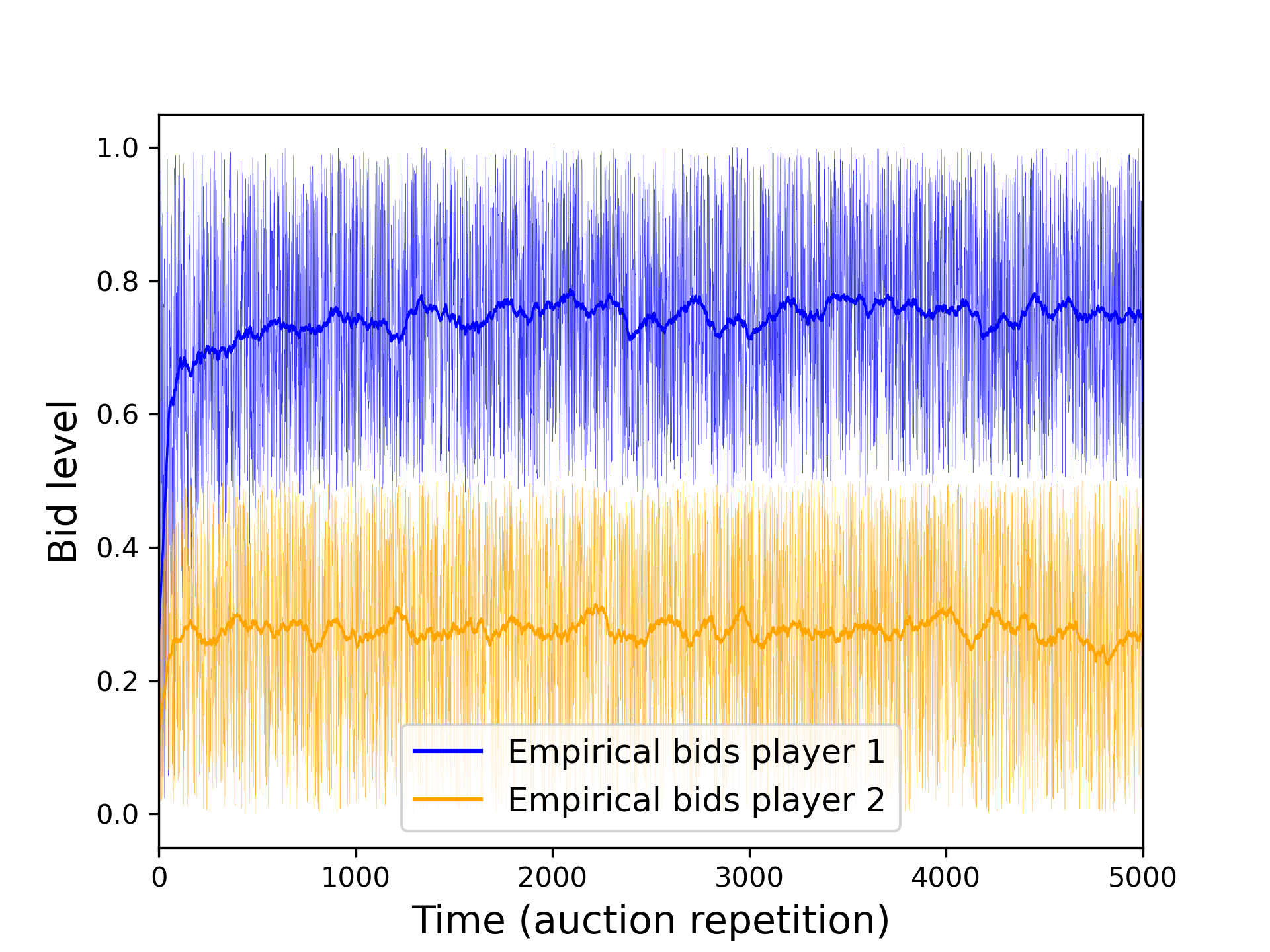}
		\caption{Hedge}  
		\label{fig:SecondPriceHedgeBidDynamics}
	\end{subfigure}
	\begin{subfigure}{.32\linewidth} 
	\center
		\includegraphics[width=1.05\linewidth]{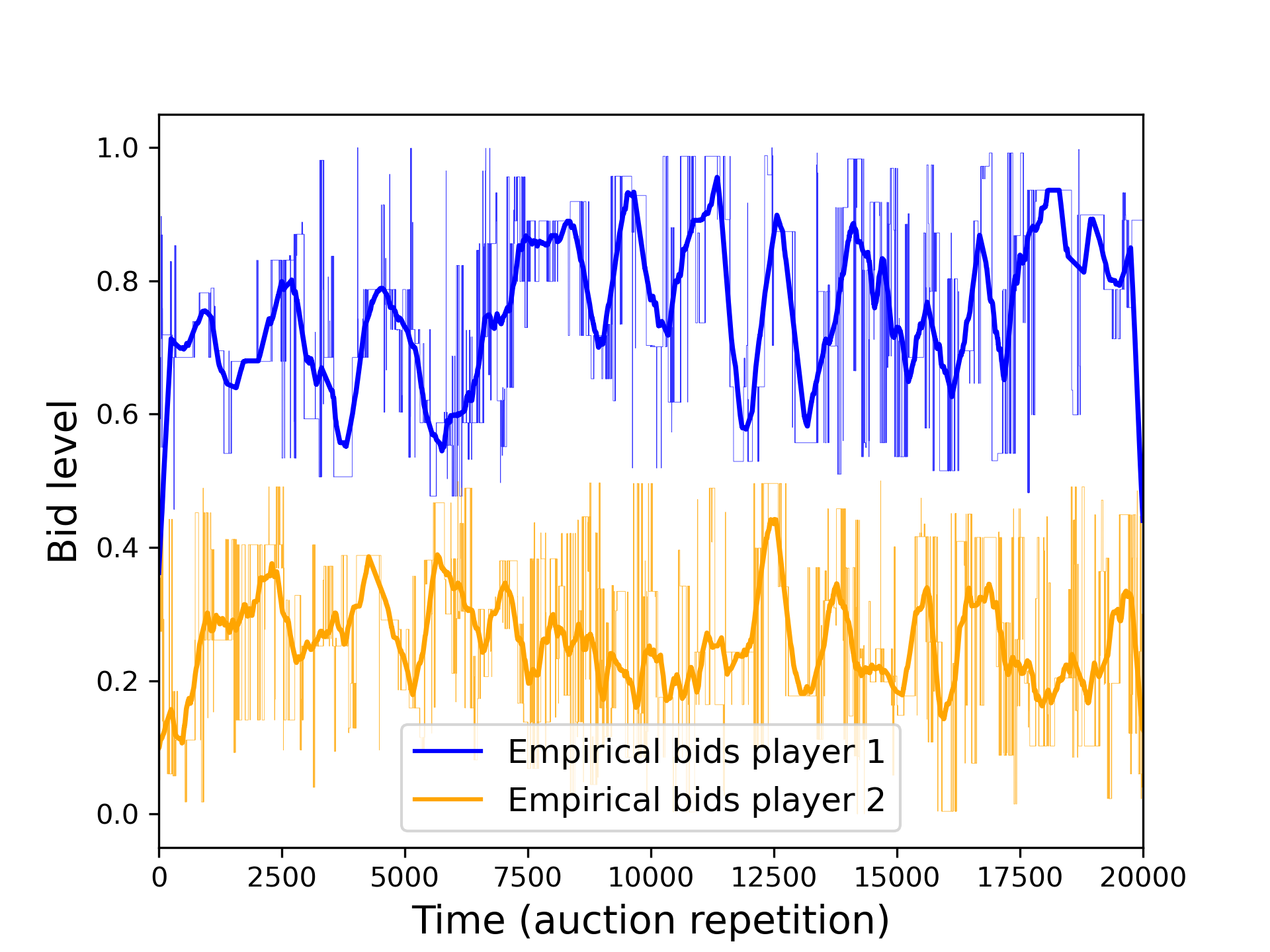}
		\caption{Recency FTPL}  
		\label{fig:SecondPriceFTPLRecencyBidDynamics}
	\end{subfigure}
	\begin{subfigure}{.32\linewidth}
		\includegraphics[width=1.05\linewidth]{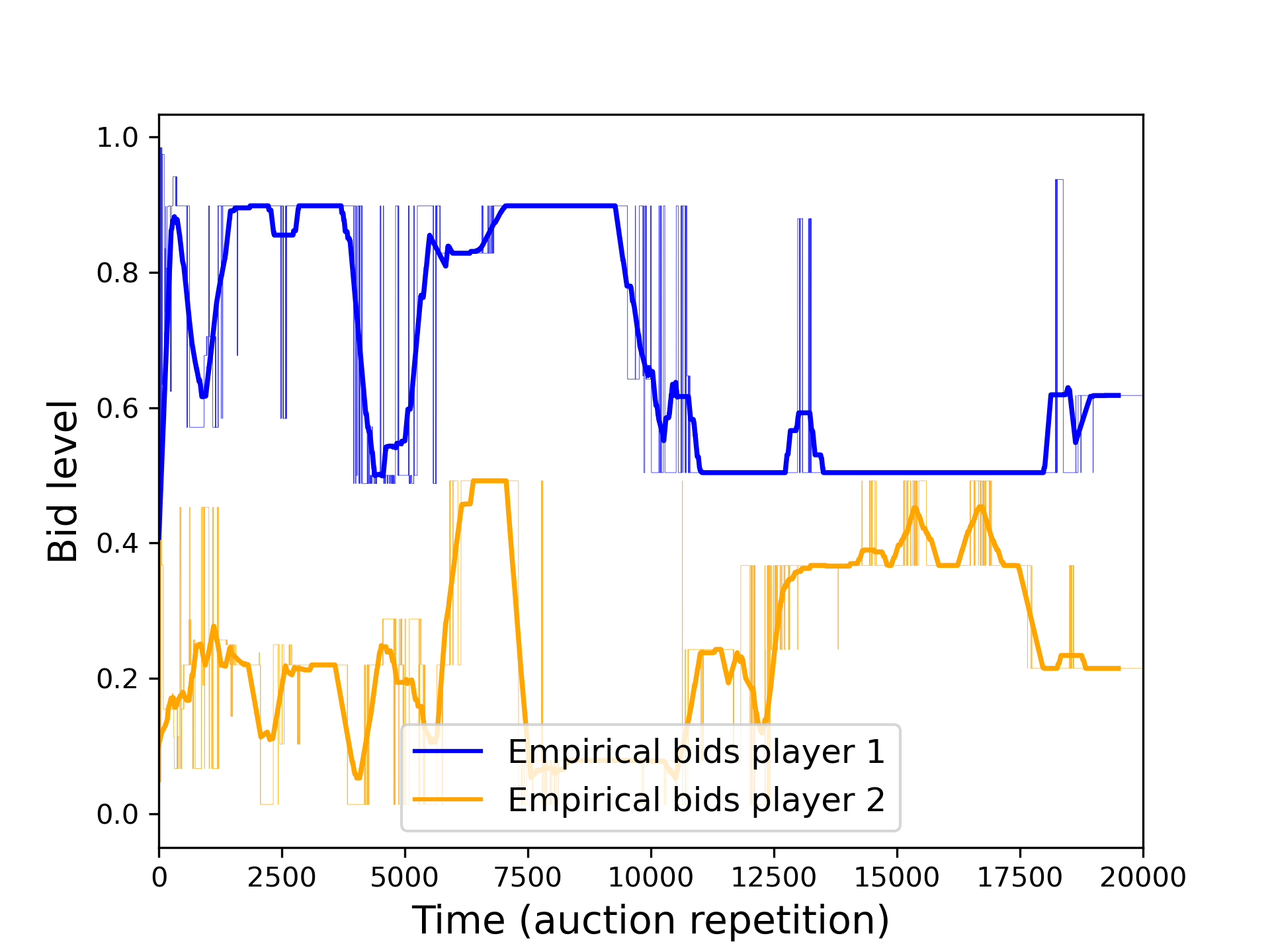}
		\caption{FTPL}  
		\label{fig:SecondPriceFTPLBidDynamics}
	\end{subfigure}
	
	\caption{{\small Bid dynamics in a repeated second-price auction with different types of regret-minimizing agents, where agent $1$ has value $1$ and agent $2$ has value $0.5$. Figure \ref{fig:SecondPriceHedgeBidDynamics}: The exponential update version of multiplicative weights (the Hedge algorithm). 
	Figure \ref{fig:SecondPriceFTPLRecencyBidDynamics}: FTPL with a recency bias.  
	Figure \ref{fig:SecondPriceFTPLBidDynamics}: FTPL.
	The solid lines show running-window averages of the bids of each player over a window of $100$ auctions in Figure \ref{fig:SecondPriceHedgeBidDynamics}, and of $500$ auctions in Figures \ref{fig:SecondPriceFTPLRecencyBidDynamics} and \ref{fig:SecondPriceFTPLBidDynamics}. 
		The average price paid by the high bidder in sequences of $50$,$000$ auctions in our simulations was $0.270 \pm 0.005$ for the Hedge algorithm, $0.26 \pm 0.11$ for FTPL, and $0.256 \pm 0.017$ for FTPL with a recency bias, compared with $0.271 \pm 0.003$ for the the linear multiplicative weights algorithm (presented in Figure \ref{fig:second-price-suctions}).
	}}
	\vspace{-9pt}
	\label{fig:second-price-Hedge-FTPLRecency}
\end{figure}

\subsection{First-Price Auctions}

\begin{figure}[!h]
\centering
\vspace{-14pt}
	\begin{subfigure}{.32\linewidth}
		\includegraphics[width=1.05\linewidth]{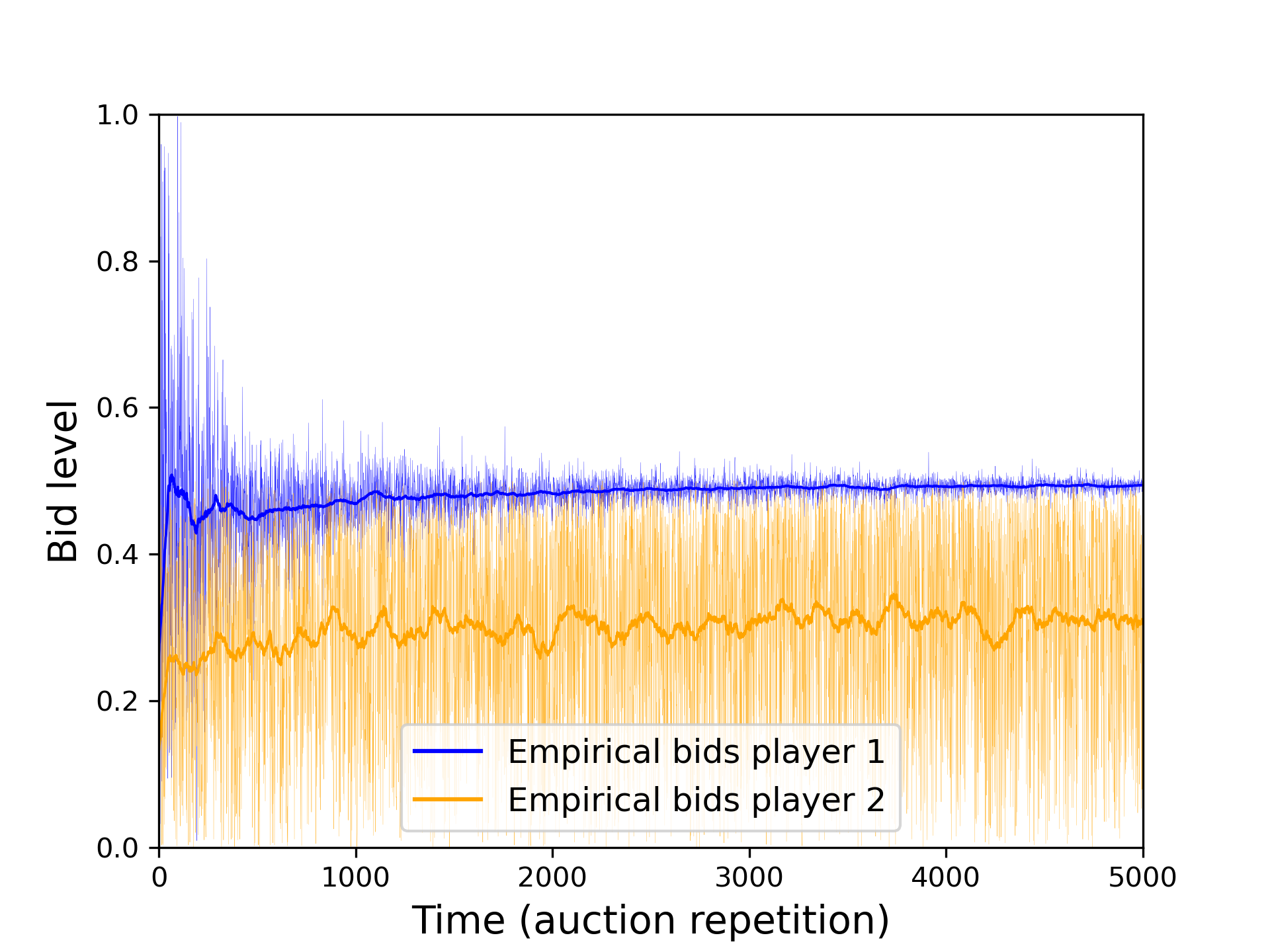}
		\caption{Hedge}  
		\label{fig:FirstPriceHedgeBidDynamics}
	\end{subfigure}
	\begin{subfigure}{.32\linewidth} 
	\center
		\includegraphics[width=1.05\linewidth]{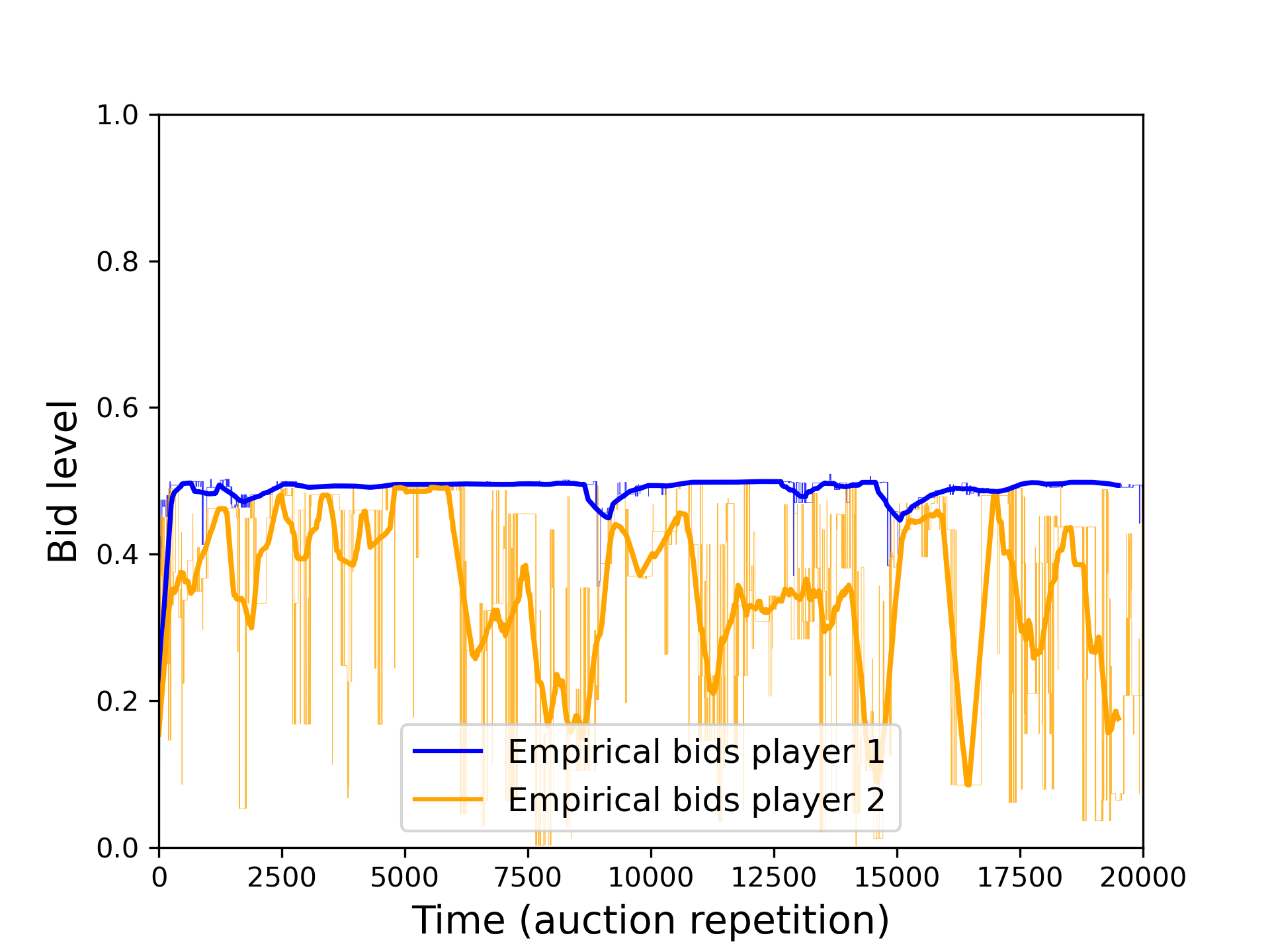}
		\caption{Recency FTPL}  
		\label{fig:FirstPriceFTPLRecencyBidDynamics}
	\end{subfigure}
	\begin{subfigure}{.32\linewidth}
		\includegraphics[width=1.05\linewidth]{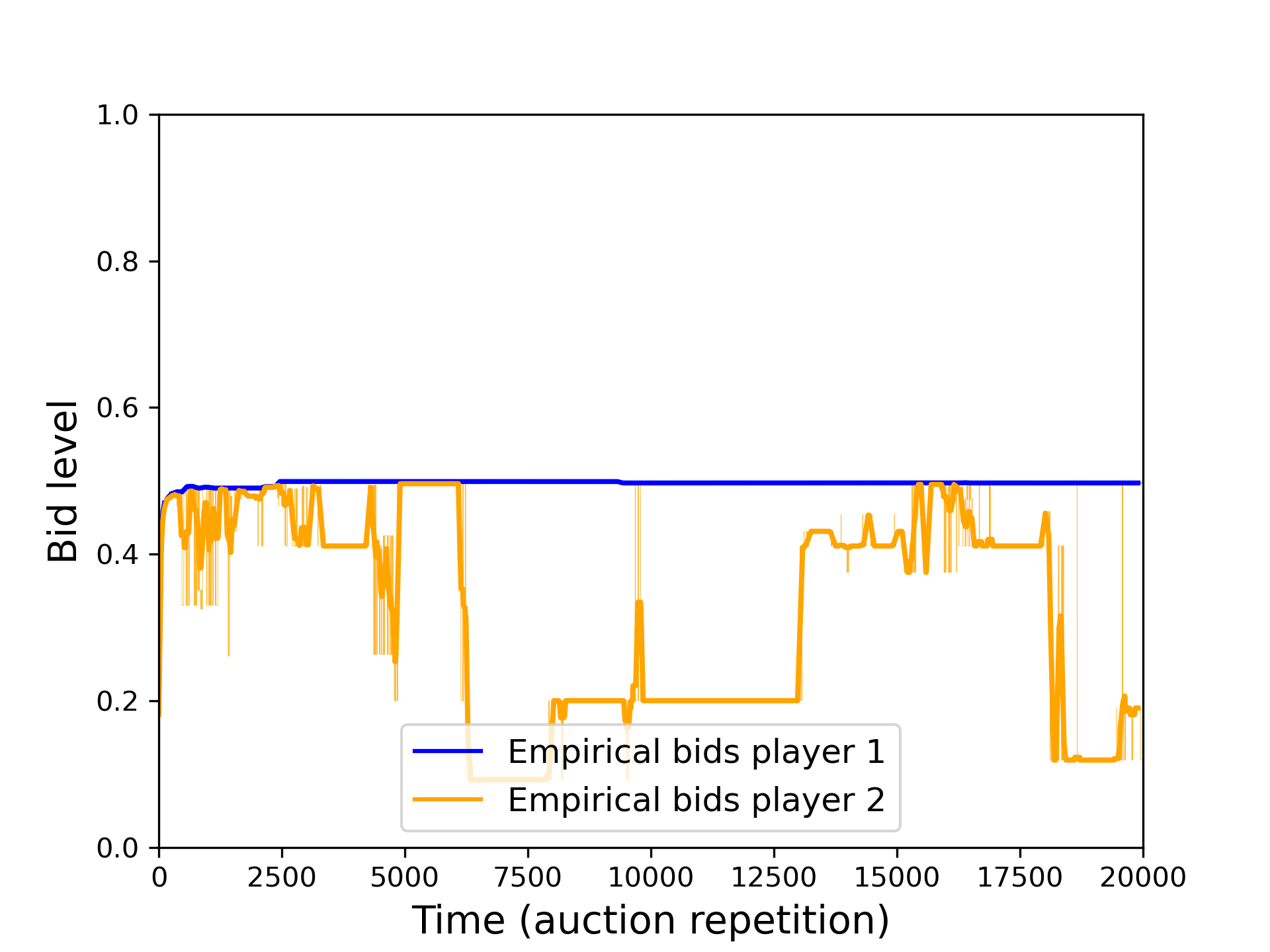}
		\caption{FTPL}  
		\label{fig:FirstPriceFTPLBidDynamics}
	\end{subfigure}
	
	\caption{{\small Bid dynamics in a repeated first-price auction with different types of regret-minimizing agents, where agent $1$ has value $1$ and agent $2$ has value $0.5$. Figure \ref{fig:FirstPriceHedgeBidDynamics}: The exponential update version of multiplicative weights (the Hedge algorithm). Figure \ref{fig:FirstPriceFTPLRecencyBidDynamics}: FTPL with a recency bias. Figure \ref{fig:FirstPriceFTPLBidDynamics}: FTPL.
	The solid lines show running-window averages of the bids of each player over a window of $100$ auctions in Figure \ref{fig:FirstPriceHedgeBidDynamics}, and of $500$ auctions in Figures \ref{fig:FirstPriceFTPLRecencyBidDynamics} and \ref{fig:FirstPriceFTPLBidDynamics}.
	}}
	\vspace{-9pt}
	\label{fig:first-price-Hedge-FTPLRecency}
\end{figure}

\section{First-Price Auctions}\label{sec:appendix-first-price}

\begin{proof} (Lemma \ref{thm:symmetric-FP-CCE-supports}):
Let $b_1, b_2$ be the lowest bids of the players in the support and let $B_1,B_2$ be the highest bids in the support. Consider the following cases: 

\vspace{3pt}
\noindent
Case I: $b_1=B_1$ (or symmetrically $b_2=B_2$).  In this case player $1$ plays a pure strategy and so the CCE condition implies that the second player must best-reply to it.  Therefore unless $b_1=B_1 \geq v-2\epsilon$, we must have $b_2=B_2=b_1+\epsilon$, which is not an equilibrium unless $b_2 = B_2 = v$.
Thus $b_1<B_1$ and $b_2<B_2$.  

\vspace{3pt}
\noindent
Case II: $b_1=b_2$. Here, if $b_1=b_2<v-2\epsilon$ then $b_1$ is strictly dominated by $b_1+\epsilon$ ($b_1$ loses unless player $2$ bids $b_2$, in which case player $1$'s utility is $(v-b_1)/2$ due to the tie, but then $v-(b_1+\epsilon)$, which this player gets from bidding $b_1+\epsilon$, is better).  Thus $b_1=b_2 \geq v-2\epsilon$.  

\vspace{3pt}
\noindent
Case III: w.l.o.g. $b_1<b_2$. Then $b_1$ always loses and always gets utility $0$. Thus, due to the co-undominated CCE condition, player $1$ must get $0$ utility for every bid in his support and every bid in player $2$'s support, and so he gets zero utility.  We thus must have $b_2=B2=v$ since otherwise a fixed bid of $v-\epsilon$ for player $1$ would beat staying in the given CCE, leading to a contradiction.
\end{proof}

\section{Generalized First-Price Auctions}\label{sec:appendix-GFP}

Here we provide further details of our analysis of GFP auctions. We state and prove Theorem \ref{thm:diagonal-CCE}, showing that the diagonal bid distribution $(x,x)$ with the Nash equilibrium marginal distributions (as were shown in Theorem \ref{thm:GFP-mixed-Nash}) is indeed a CCE of the auction game. 
We then continue with the formal details of our analysis of the  
Nash equilibria of the meta-game between the users of regret-minimizing agents that repeatedly play the auction for the nearly diagonal bid distribution.

\subsection{Diagonal Coarse Correlated Equilibrium}

Next we formally show, in Theorem \ref{thm:diagonal-CCE}, that the diagonal distribution discussed in the introduction is indeed a CCE. In this equilibrium both bidders place equal bids in every auction that are distributed according to the Nash density function {\small$F(x) = \frac{x}{1-x}$}. We begin with a lemma that formalizes the CCE condition for the case of diagonal distributions $(x,x)$, and then proceed to the proof of the theorem.

\begin{lemma}\label{thm:diagonal-CCE-condition} 
Consider the GFP auction with two items with click-through rates $(1,1/2)$ played by two players with values-per-click $(v,w)$. Assume w.l.o.g. that $v \geq w$, and assume that ties are broken uniformly at random. A diagonal bid distribution $(x,x)$, with support $[0,w/2]$ and density $f(x)$ and cumulative distribution $F(x)$, is a coarse correlated equilibrium of the auction game if and only if the following two inequalities hold:
\footnotesize
$
\frac{3}{4} \int_0^{w/2} f(x)(v-x) dx \geq \frac{1}{2} \max_{x_0 \in [0,w/2]} \big(F(x_0) + 1 \big)(v-x_0)
$, and 
$
\frac{3}{4} \int_0^{w/2} f(x)(w-x) dx \geq \frac{1}{2} \max_{x_0 \in [0,w/2]} \big(F(x_0) + 1 \big)(w-x_0)
$.
\normalsize
\end{lemma}

\begin{proof} 
First, in the case where both players play according to $f$, i.e., the bids are $(x,x)$ where $x\sim f(x)$, the utility of the player with value $v$ from bidding $x$ is $3/4\cdot(v-x)$ and the utility of the player with value $w$ is $3/4\cdot (w-x)$. The expected utilities are then given by the integrals on the left-hand side of the two inequalities. 

In a coarse correlated equilibrium these expected utilities, of the players must be greater than or equal to the maximum expected utility that each player could obtain by playing a fixed bid, assuming that the other player bids according to the distribution $f$. Notice that for the player with the high value $v$, any bid $x_0>w/2$ gives lower utility than bidding $w/2$ (assuming that the other player plays according to $f$), and for the player with the low value $w$, any bid $x_0>w/2$ is strictly dominated by bidding 0. Therefore, it is sufficient to look only at fixed bids $x_0 \leq w/2$ for both players.

The expected utility of the player with value $v$ from bidding a fixed bid {\small$x_0$ is: $F(x_0)\cdot (v-x_0) + (1-F(x_0))\cdot (v-x_0)/2 = 1/2 \cdot \big(F(x_0) + 1 \big)(v-x_0)$}, and the expected utility of the player with value $w$ from bidding a fixed bid {\small$x_0$} {\small is $F(x_0)\cdot (w-x_0) + (1-F(x_0))\cdot (w-x_0)/2 = 1/2 \cdot \big(F(x_0) + 1 \big)(w-x_0)$}. Taking the maximum, we obtain the expressions of the right-hand side of the two inequalities. 
\end{proof}

\begin{theorem}\label{thm:diagonal-CCE} 
Consider a GFP auction as in Lemma \ref{thm:diagonal-CCE-condition}, with {\small$v$=$w$=$1$}.
The diagonal distribution {\small$(x,x)$} with support {\small$[0,1/2]$}, and {\small$F(x)=\frac{x}{1-x}$} (and density {\small$f(x) = \frac{1}{(1-x)^2}$}),  
is a CCE of the auction, with an expected payoff of {\small$\frac{3}{4} \ln(2) \cong 0.52$} to each player. Additionally, there are other diagonal CCEs  with same support and different expected payoffs.
\end{theorem}

\begin{proof} 
By Lemma \ref{thm:diagonal-CCE-condition}, the CCE condition for the symmetric case {\small$v=w=1$}, is:
\footnotesize
$
\frac{3}{4} \cdot \int_0^{1/2}{f(x) (1-x)} dx \geq \frac{1}{2} \cdot \max_{x_0 \in [0,1/2]} (1-x_0)\big(F(x_0)\big) + 1).
$
\normalsize
This condition is identical for each of the players. 
Integrating the left-hand side of the inequality and substituting {\small $F(x) = \frac{x}{1-x}$} in the right-hand side, we obtain 
\footnotesize
$
\frac{3}{4} \cdot \ln(2) \geq \frac{1}{2} \cdot \max_{x_0 \in [0,1/2]} (1-x_0)\big(\frac{x_0}{1-x_0}\big) + 1) = \frac{1}{2}.
$
\normalsize
That is, with the distribution $F$, the right-hand side is constant: any fixed bid played against this distribution gives a utility of $1/2$. Since $3/4\cdot \ln(2) = 0.519... > 1/2$, the inequality holds for every fixed bid.

Other diagonal CCEs with the same support can be shown in a similar way.   
For example, the diagonal uniform distribution $f(x) = 2$ is a CCE with an expected payoff of $9/16$, and the diagonal linear distribution $f(x) = 8x$ is a CCE with an expected payoff of $1/2$. These are obtained directly by substituting the distributions into the inequality stated above. 
\end{proof}

\subsection{Equilibrium of the Meta-Game}\label{app:GFP-meta-game-NE}
Next, we turn to consider the meta-game of the auction game, and to derive the equilibrium of this meta-game. 
The declared values provided by the players and used by the agents are denoted by $v$ and $w$, and we assume w.l.o.g. that $v \geq w$, and in cases where it is needed to check whether a player would prefer to declare otherwise we will do so explicitly.   
The analysis is based on the assumption that the agents converge to the nearly diagonal distribution discussed in the introduction where the ``high agent,'' with value $v$, places bid $x$ distributed according to the Nash distribution $F=\frac{x}{1-x}$, and the other agent bids $x$ with probability $w/v$ and otherwise bids $0$. This is validated empirically, as discussed in Section \ref{sec:GFP}. 
We begin with the definition of a nearly diagonal distribution.
\begin{definition}
A distribution $D$ on $(x,y)$ is called \emph{nearly diagonal} if its support is on the diagonal, $x=y$, and on the $x$-axis, $y=0$.  
\end{definition}
 
The following lemma formalizes our prediction for the empirical play of two agents with values $v>w$, and shows the utilities obtained by players with equal values who misreport their values to their agents as $v,w$ under convergence to this distribution.

\begin{lemma}\label{thm:symmetric_GFP_nearly_diagonal_utilities}
Consider the two-player GFP auction with two slots with click-through rates {\small$(1,1/2)$},  
where the true value-per-click for each of the two players is $1$. The values that users declare to their agents are {\small$v>0$} for player $1$ and {\small$w>0$} for player $2$, and 
assume w.l.o.g. that {\small$v \geq w$}.  
In a nearly agonal distribution {\small$(x,y)$} where {\small$x\sim F(x)=\frac{x}{1-x}$} with support {\small$[0,w/2]$} and {\small$y=x$} w.p. {\small$p=w/v$} and {\small$y=0$} otherwise, the expected utilities of the players are 
\footnotesize
$
u_1(v,w) = \big(1-w(1-\ln(2))\big)\Big(1-\frac{w}{4v}\Big), \quad u_2=\frac{1}{2}\Big(1-\frac{w}{v}\Big) + \frac{3w}{4v}\big(1-w(1-\ln(2))\big)
$.
\normalsize
\end{lemma}

\begin{proof} 
Player $2$ bids zero with probability $1-p$; in these cases, player $1$ wins the top slot and gets a utility of $v-x$.  In the other cases (i.e., with probability $p$), both players bid the same, and player $1$ has a utility of $3/4 \cdot (v-x)$. The expected utility of player $1$ is thus
\vspace{-6pt}
\footnotesize
\begin{equation*}
\vspace{-6pt}
u_1(v,w) = \Big(1-\frac{p}{4}\Big)\int_0^{\frac{w}{2}} f(x)(1-x) dx = \Big(1-\frac{p}{4}\Big)\big(1-w(1-\ln(2))\big).
\end{equation*} 
\normalsize
Player $2$ bids zero with probability $1-p$ and has in these cases a utility of $w/2$, or bids $x$ with probability $p$ and then has a utility of $3/4\cdot (w-x)$. The expected utility to player $2$ is thus
\vspace{-6pt}
\footnotesize
\begin{equation*}
\vspace{-6pt}
u_2(v,w) = \frac{(1-p)w}{2} + \frac{3p}{4} \int_0^{\frac{w}{2}} f(x)(1-x) dx = \frac{(1-p)w}{2} + \frac{3p}{4} \big(1-w(1-\ln(2))\big).
\end{equation*} 
\normalsize
Substituting $p=w/v$ in the expressions we obtained for $u_1$ and $u_2$ gives the result stated.
\end{proof}

The next result specifies the equilibria of the meta-game, in which users select values to report to their regret-minimizing agents.  
\begin{proposition}\label{thm:GFP-meta-game-equilibria} 
Consider the GFP auction as in Lemma \ref{thm:symmetric_GFP_nearly_diagonal_utilities}. 
Denote the declared values of agents $1$ and $2$ by $v$ and $w$, respectively. Assume w.l.o.g. that $v \geq w$, and that the agents of these players converge to the nearly diagonal distribution $(x,y)$ as in Lemma \ref{thm:symmetric_GFP_nearly_diagonal_utilities}.

\vspace{5pt}
\noindent
If there is a maximum value declaration {\small$M > \frac{1}{1-\ln(2)}$}, 
then the meta-game with these dynamics has two Nash equilibria in which the 
 players declare {\small$(\frac{1}{6(1-\ln(2))}, M)$} or {\small$(M, \frac{1}{6(1-\ln(2))})$}. If the value declarations are not bounded, then for any $\epsilon>0$ every declaration profile in 
which one player declares {\small$\frac{1}{6(1-\ln(2))}$} and the other player declares a value $h>1/\epsilon$ is an $\epsilon$-Nash-equilibrium of the meta-game.
\end{proposition}

\begin{proof} 
According to Lemma \ref{thm:symmetric_GFP_nearly_diagonal_utilities}, 
in the nearly diagonal distribution defined in the proposition, the expected utilities of the players (assuming $v\geq w$) are 
\footnotesize
\begin{equation*}
u_1(v,w) = \big(1-w(1-\ln(2))\big)\Big(1-\frac{w}{4v}\Big), \ u_2=\frac{1}{2}\Big(1-\frac{w}{v}\Big) + \frac{3w}{4v}\big(1-w(1-\ln(2))\big).
\end{equation*}
\normalsize
To find an equilibrium, we first require that the utility to player $2$ be maximal:
{\small 
$
\frac{\partial u_2}{\partial w} = 
\frac{1}{4v} 
\big(1-6w(1-\ln(2))\big) = 0
$. 
}
A maximum is obtained at {\small$w=\frac{1}{6(1-\ln(2))} \cong 0.543$}.
Next, we look at the derivative for the ``high player'': 
{\small$\frac{\partial u_1}{\partial v} = 
\frac{w}{4v^2} \big(1-w(1 - \ln(2))\big)$}. As it turns out, this derivative is positive for every 
{\small$w < \frac{1}{1-\ln(2)}$}, and specifically, it is positive also for 
{\small$w=\frac{1}{6(1-\ln(2))}$}. 
Therefore, if one player declares {\small$\frac{1}{6(1-\ln(2))}$}, then the other player prefers to declare as high a value as possible. 

Let {\small$M>\frac{1}{1-\ln(2)}$} and assume (w.l.o.g. about switching the names of players $1$ and $2$) that player $1$ declares $v=M$ and that player $2$ declares {\small$w=\frac{1}{6(1-\ln(2))}$}. 
We have shown already that player $1$ has lower utility if he declares any declaration lower than $M$ that is still above $w$ (since the derivative of the utility is  positive), and that player $2$ has less utility if he declares any other declaration that is still less than or equal to $v$ (since this is not his maximum point). We now need to show that no player prefers to unilaterally switch positions, i.e., that player $1$ would not prefer to declare below $w$, and that player $2$ would not prefer to declare above $M$. This can be seen directly from the utility functions derived above.  
On the one hand, when player $1$ declares  
$v=M$ 
(and player $2$ declares $w$), player $1$ has utility 
{\small$u_1(v,w)=\frac{5}{6}\big(1-\frac{1}{24M(1-\ln(2))}\Big) > 
\frac{115}{144} > \frac{3}{4}$}. 
On the other hand, if player $1$ declares $w$ or less (thus taking the role of the low player), then he has utility at most $\frac{5}{8}$. 
Hence the high player (player $1$) would not prefer to declare less than $M$, whether this is below or above $w$ (but he would prefer to declare values higher than $M$ if this is possible, as shown above). 

Next, when player $2$ declares {\small$w=\frac{1}{6(1-\ln(2))}$}, he has utility more than $1/2$. If player $2$ declares above player $1$, i.e., above $M$, he has negative utility and therefore he would not prefer to take the role of the high player, and his best reply, given that player $1$ declares $v=M$, is {\small$w=\frac{1}{6(1-\ln(2))}$}. If $M$ is the maximum possible declaration in the meta-game, then $v=M$ is also the best reply of player $1$, given that player $2$ declares {\small$w=\frac{1}{6(1-\ln(2))}$}. That is, when there is a maximum declaration {\small$M>\frac{1}{1-\ln(2)}$}, the two declaration profiles in which one player declares $M$ and the other declares {\small$\frac{1}{6(1-\ln(2))}$} are Nash equilibria of the meta-game.     

If declarations are not bounded from above, let $\epsilon>0$, and let {\small$M=\frac{1}{24\epsilon(1-\ln(2))}$}. Assume w.l.o.g. that player $1$ is the high player and declares $v=M$, and player $2$ is the low player and declares {\small$w=\frac{1}{6(1-\ln(2))}$}. We have already shown that in this profile player $2$ is best-replying to player $1$, but player $1$ would prefer to increase his declaration. To show an $\epsilon$-equilibrium, we claim that player $1$ cannot increase his utility by more than $\epsilon$ from any other declaration. 
From the expression of the utility stated above it follows that the utility of player $1$ from declaring $v=M$, assuming player $2$ does not change $w$, is $u_1 = \frac{5}{6}\big(1-\frac{1}{24M(1-\ln(2))}\big) = \frac{5}{6}\big(1-\epsilon\big)$. Since $u_1$ is monotonically increasing with $v$ (given that $w$ remains fixed) and in the limit $v\rightarrow \infty$ player $1$'s utility is $5/6$, player $1$ cannot increase his gains (compared with his gains when declaring $v=M$) by more than $\epsilon$ with any declaration. Therefore, every declaration profile {\small$w=\frac{1}{6(1-\ln(2))}$} and {\small$v \geq \frac{1}{24\epsilon(1-\ln(2))}$} (and specifically, also $v>1/\epsilon$), or vice versa, is an $\epsilon$-Nash equilibrium of the meta-game.
\end{proof}


\end{document}